\documentclass[lettersize,journal]{IEEEtran}
\usepackage{array}
\usepackage{graphicx}    
\usepackage{epsfig}
\usepackage{cite}
\usepackage{epstopdf}
\usepackage{amsfonts,amsmath,amsthm,amssymb,balance}
\usepackage{multirow,balance}
\usepackage{bm}
\usepackage[usenames,dvipsnames]{color}
\usepackage{colortbl}
\usepackage{float}
\usepackage{subfigure}
\usepackage{amsmath}
\usepackage{lipsum}
\usepackage{amsmath,epsfig,amssymb,subfigure,bm,dsfont}
\usepackage{stfloats}
\usepackage{mathrsfs}
\usepackage{caption}
\usepackage{epsf}
\usepackage{psfrag}
\usepackage{graphicx}
\usepackage{amssymb}
\usepackage{multirow}
\usepackage{booktabs}
\usepackage{algorithmicx}
\usepackage{algpseudocode}
\usepackage{tabularx}
\usepackage{exscale}      
\usepackage{relsize}      
\usepackage{verbatim}     
\usepackage{CJK}
\usepackage{indentfirst}
\usepackage[flushleft]{threeparttable}
\usepackage{color}         
\usepackage{amssymb}       
\usepackage{amsmath}       
\usepackage[linesnumbered, ruled]{algorithm2e} 
\usepackage{graphicx} 
\usepackage[implicit=false]{hyperref}
\usepackage{threeparttable}

\usepackage{longtable}
\usepackage{booktabs}
\usepackage{tabularx}

\newcommand{\orcid}[1]{\href{https://orcid.org/#1}{\includegraphics[width=10pt]{ORCIDiD128.png}}}
\makeatletter
\renewcommand{\maketag@@@}[1]{\hbox{\m@th\normalsize\normalfont#1}}%
\makeatother
\definecolor{mygray}{gray}{.9}

\newtheorem{lemma}{Lemma}
\newtheorem{Define}{Definition}
\newtheorem{proposition}{Proposition}
\graphicspath{{figures/}}
\newcolumntype{C}[1]{>{\PreserveBackslash\centering}p{#1}}
\newcolumntype{R}[1]{>{\PreserveBackslash\raggedleft}p{#1}}
\newcolumntype{L}[1]{>{\PreserveBackslash\raggedright}p{#1}}
\ifCLASSINFOpdf
\else
\fi
\begin{document}
\title{Ternary Stochastic Geometry Theory for Performance Analysis of RIS-Assisted UDN}
\author{Hong-chi Lin, ~\IEEEmembership{IEEE Student~Member},
        and Qi-yue~Yu$^\dagger$, ~\IEEEmembership{IEEE Senior~Member}
		\thanks{H.-C.~Lin (email: linhongchi@stu.hit.edu.cn) and Q.-Y.~Yu (email: yuqiyue@hit.edu.cn) are with the Communication Research Center, Harbin Institute of Technology, China. 
		}
		\thanks{The work presented in this paper was supported by the National Natural Science Foundation of China under Grand No. 62071148.}
}
\markboth{IEEE TRANSACTIONS ON COMMUNICATIONS,~Vol.~XX, No.~X, May~2024}%
{Shell \MakeLowercase{\textit{et al.}}: Bare Demo of IEEEtran.cls for IEEE Journals}
\maketitle

{\color{blue}
\begin{abstract}
	Currently, network topology becomes increasingly complex with the increased number of various network nodes, bringing in the challenge of network design and analysis. 
	Most of the current studies are deduced based on the binary system stochastic geometry, overlooking the coupling and collaboration among nodes. 
	This limitation makes it difficult to accurately analyze network systems, such as reconfigurable intelligent surface (RIS) assisted ultra-dense network (UDN).
	To address this issue, we propose a dual coordinate system analysis method, by using dual observation points and their established coordinates.
	The concept of a typical triangle that consists of a base station (BS), a RIS, and a user equipment (UE) is defined as the fundamental unit of analysis for ternary stochastic geometry. 
	Furthermore, we extend Campbell's theorem and propose an approximate probability generating function for ternary stochastic geometry.
	Utilizing the theoretical framework of ternary stochastic geometry, we derive and analyze performance metrics of a RIS-assisted UDN system, such as coverage probability, area spectral efficiency, area energy efficiency, and energy coverage efficiency.
	Simulation results show that RIS can significantly enhance system performance, particularly for UEs with high signal-to-interference-plus-noise ratios, exhibiting a phenomenon similar to the Matthew effect. 
\end{abstract}}

\IEEEpeerreviewmaketitle

\section{Introduction}

{\color{blue}
As the mobile data traffic keeps increasing at an exponential rate, there are a large number of access points with different types. 
To support diverse communication scenarios, various network architectures have been proposed, such as user-centered \cite{User_Centric}, cell-free \cite{Cell_Free_Massive_MIMO} and reconfigurable intelligent surface (RIS) assisted networks \cite{jia1}.
Evidently, different types of access points and complex network architectures bring in challenges, including resource overload, operational complexities, and difficulties in network analysis and design. 

Traditional network analysis has evolved from single-point-to-single-point models, such as classical ultra-dense network (UDN), to multipoint-to-multipoint models, such as multiple input multiple output (MIMO), heterogeneous network (HetNet) and relay. 
Nevertheless, the presence of a single type of transceiver node in traditional network analysis results in independent communication links. 
From a geometric perspective, the current network analysis primarily focuses on linear structures. 
The advantage of binary analysis lies in its ability to simplify system performance analysis by solely considering the correlation distance. 
However, this approach also imposes limitations on accurately representing the coupling, collaboration, and auxiliary communication among multiple types of nodes.
To effectively analyze complex networks, it becomes imperative to transition from  one-dimensional linear structure (binary system) to a more comprehensive  two-dimensional surface structure (ternary system).
This transition is particularly crucial in the context of RIS-assisted networks, which hold great promise for the future wireless communication \cite{Towards_Smart_and_Reconfigurable_Environment,New_Trends_SG}.

An efficient performance analysis manner for random networks is through stochastic geometry, which can provide a tight lower bound of the actual network simulation from a stochastic analysis perspective \cite{SG_tutorial}. 
Over the past decade, there has been significant growth in the analysis of wireless networks using stochastic geometry techniques \cite{Stochastic_Geometry_for_Wireless_Networks}.
In particular, the Poisson point process (PPP) has been proposed as a special point process to model the spatial distribution of transceivers in wireless networks. 
The PPP model highly accurate in describing the unique characteristics of dense networks. 
This modeling approach has provided valuable insights for network design and optimization.
Stochastic geometric approaches represent the later and more successful addition to a lineage of analytical approaches that have been developed since the inception of cellular models for wireless communication \cite{Stochastic_geometry_and_its_applications}.

Many studies have utilized stochastic geometry to investigate MIMO, heterogeneous networks, and relay networks, which are based on the classical binary systems with transceivers operating on the same resource \cite{MIMO_hetnet, SG_UAV_liuyuanwei,SG_relay_system,Multicell_MIMO_Communications_Relying,jia2,jia3,Modeling_and_Analysis_of_K_Tier,SG_multi-tier,SG_heterogeneity_jinshi}.
In the case of MIMO systems, signal processing techniques enable the equivalence of multiple parallel independent channels. This means that the communication process in MIMO does not necessitate the consideration of distances between antennas or between user equipments (UEs). As a result, MIMO systems can be perceived as multi-parallel binary systems \cite{MIMO_hetnet,jia3,Multicell_MIMO_Communications_Relying}.
Additionally, in heterogeneous networks, various types of BSs coexist, such as macro and micro BSs. 
However, the communication process in such networks always ignore the effect of distance between BSs. 
Consequently, it can be regarded as a multi-layer binary system \cite{Modeling_and_Analysis_of_K_Tier,SG_multi-tier,SG_heterogeneity_jinshi}.
For relay networks, the communication between the BS and the relay, as well as the relay and the UE, is typically performed independently in a time-sharing manner. Besides the relay selection, the relay network communication process also disregards the distance between the BS and the UE, thus enabling its interpretation as a multi-segment binary system \cite{SG_UAV_liuyuanwei,SG_relay_system,jia2}.
When it comes to collaborative analysis involving multiple function nodes, the current binary analysis framework of stochastic geometry struggles to adequately represent the system performance in a reasonable and accurate manner.

Regarding the ternary analysis model, a notable application is to analyze the RIS-assisted UDN system. 
The distinctive working principle of RIS necessitates the collaboration of both direct and reflected links to facilitate communication. Moreover, from the perspective of RIS-assisted network technology itself, the favorable characteristics of RIS (i.e., low cost, low power consumption, and easy installation) enable improvements in network capacity and coverage.
Consequently, it is required to investigate the correlation between the dual links to provide accurate network analysis \cite{New_Trends_SG}.
Numerous articles have been published on RIS-assisted UDN systems using stochastic geometry \cite{SG_RIS_CELLULAR,coverage_probability_ris_assisted,RIS_spatially,interference_RIS_intelligent,Large_Scale_Deployment_of_Intelligent_Surfaces,RIS_Assisted_Coverage_Enhancement,Hybrid_Active/Passive_Wireless,Spatial_Throughput_Characterization,wireless_network_RIS,SG_RIS_assisted,SG_fine_jinshi}.}
Initially, these studies assumed that the direct and reflected links are completely separated by different working modes \cite{coverage_probability_ris_assisted} and orthogonal time-frequency resources \cite{SG_RIS_CELLULAR}. 
In \cite{RIS_spatially} and \cite{interference_RIS_intelligent}, researchers applied the line of sight (LoS) indicator function to model the relationship between the direct and reflected links. 
The line Boolean blockages model is employed in \cite{Large_Scale_Deployment_of_Intelligent_Surfaces} to analyze the direct and reflected cases.
Subsequently, it is assumed that the direct and reflected links might be independent from a receive diversity perspective \cite{RIS_Assisted_Coverage_Enhancement}. Based on the statistical analysis approximations, the direct and reflected links are considered roughly independent \cite{Spatial_Throughput_Characterization,Hybrid_Active/Passive_Wireless}. 
In \cite{SG_RIS_assisted}, the authors simplify the system analysis by treating the BS, RIS, and UE as an equilateral triangle. The successful reflection probability and distance distribution are calculated using a line segment object model in \cite{wireless_network_RIS}. 
Lastly, \cite{SG_fine_jinshi} incorporate the reflection probability into measurement analysis using triangulation modeling.

\begin{table*}[t]
 	\begin{center}
			\label{UDN_vs_RIS}
 	      \caption{Comparisons between binary stochastic geometric  and  ternary stochastic geometric.}
 	\begin{tabular}{|c |c |c|c|}
		\toprule
		\specialrule{-0.05em}{1pt}{1pt}
		 \cline{1-4}
	\cline{1-4}
 	\multicolumn{2}{c|}{ }                                                             & \multicolumn{1}{c|}{Binary stochastic geometric}                   & \multicolumn{1}{c}{Ternary stochastic geometric}   \\ 
	 \cline{1-4}
 	\multicolumn{1}{c|}{\multirow{5}{*}{Theoretical basis}} & \multicolumn{1}{c|}{{Voronoi diagram }}                     & \multicolumn{1}{c|}{ones/two point process}                   & \multicolumn{1}{c}{ multiple point process}         \\ 
	 \cline{2-4}
 	\multicolumn{1}{c|}{}                    & \multicolumn{1}{c|}{reference point}                    & \multicolumn{1}{c|}{typical UE/BS}                  & \multicolumn{1}{c}{typical triangle}    \\ \cline{2-4}
 	\multicolumn{1}{c|}{}                    & \multicolumn{1}{c|}{Campbell's theorem}                   & \multicolumn{1}{c|}{\cite{SG_tutorial,SG_multi-tier,Stochastic_Geometry_for_Wireless_Networks,Stochastic_geometry_and_its_applications }}                    & \multicolumn{1}{c}{(\ref{Campbell})}        \\ 
	 \cline{2-4}
 	\multicolumn{1}{c|}{}                    & \multicolumn{1}{c|}{ probability generating functional }                     & \multicolumn{1}{c|}{\cite{SG_tutorial,SG_multi-tier,Stochastic_Geometry_for_Wireless_Networks,Stochastic_geometry_and_its_applications } }                    & \multicolumn{1}{c}{(\ref{PGFL})}        \\ \cline{1-4}
 	\multicolumn{1}{c|}{\multirow{1}{*}{System model}} & \multicolumn{1}{c|}{\multirow{1}{*}{signal transmission}} & \multicolumn{1}{c|}{\multirow{1}{*}{\cite{SG_tutorial,SG_multi-tier,SG_UAV_liuyuanwei,Stochastic_Geometry_for_Wireless_Networks,Stochastic_geometry_and_its_applications }}} & \multicolumn{1}{c}{(\ref{y_u}), (\ref{betaeq1delta}) }       \\ \cline{1-4}
 	\multicolumn{1}{c|}{\multirow{3}{*}{Application}} & \multicolumn{1}{c|}{signal statistics}                 & \multicolumn{1}{c|}{ \cite{SG_tutorial,SG_multi-tier,Stochastic_Geometry_for_Wireless_Networks,Stochastic_geometry_and_its_applications,wireless_network_RIS }}                    & \multicolumn{1}{c}{(\ref{x_mean}), (\ref{I_mean})}       \\ \cline{2-4}
 	\multicolumn{1}{c|}{}                    & \multicolumn{1}{c|}{cover probability/ASE/AEE}                   & \multicolumn{1}{c|}{ \cite{SG_tutorial,SG_multi-tier,SG_relay_system,Modeling_and_Analysis_of_K_Tier,SG_heterogeneity_jinshi,SG_UAV_liuyuanwei,Stochastic_Geometry_for_Wireless_Networks,Stochastic_geometry_and_its_applications,AEE }}                    & \multicolumn{1}{c}{ (\ref{coverage_hebing})/(\ref{ase_frist})/(\ref{eta_e})}       \\ \cline{2-4}  
	 \multicolumn{1}{c|}{} &  \multicolumn{1}{c|}{ECE} & \multicolumn{1}{c|}{/}                    & \multicolumn{1}{c}{ (\ref{eta_c})}       \\
	 \cline{1-4}
 	\multicolumn{1}{c|}{ \multirow{5}{*}{Examples $\&$ Categorizations}  }                                                           & \multicolumn{1}{c|}{classical UDN}         &     \multicolumn{1}{c}{\cite{SG_tutorial,Stochastic_Geometry_for_Wireless_Networks,Stochastic_geometry_and_its_applications }}     & \multicolumn{1}{c}{ /  }   \\ 
	 \cline{2-4}
	 \multicolumn{1}{c|}{ }        & \multicolumn{1}{c|}{MIMO (multi-parallel binary system)}       &      \multicolumn{1}{c|}{\cite{MIMO_hetnet,jia3,Multicell_MIMO_Communications_Relying}.}      &  \multicolumn{1}{c}{/} \\
	 \cline{2-4}
	 \multicolumn{1}{c|}{ }        & \multicolumn{1}{c|}{HetNet (multi-layer binary system)}         &     \multicolumn{1}{c|}{\cite{Modeling_and_Analysis_of_K_Tier,SG_multi-tier,SG_heterogeneity_jinshi}}     & \multicolumn{1}{c}{/} \\
	 \cline{2-4}
	 \multicolumn{1}{c|}{ }        & \multicolumn{1}{c|}{relay (multi-segment binary system)}      &       \multicolumn{1}{c|}{\cite{SG_UAV_liuyuanwei,SG_relay_system,jia2}}      & \multicolumn{1}{c}{/} \\
	 \cline{2-4}
	 \multicolumn{1}{c|}{ }        & \multicolumn{1}{c|}{RIS-assisted UDN}       &    \multicolumn{1}{c|}{\cite{SG_RIS_CELLULAR,coverage_probability_ris_assisted,RIS_spatially,interference_RIS_intelligent,Large_Scale_Deployment_of_Intelligent_Surfaces,RIS_Assisted_Coverage_Enhancement,Hybrid_Active/Passive_Wireless,Spatial_Throughput_Characterization,wireless_network_RIS,SG_RIS_assisted,SG_fine_jinshi}}        & \multicolumn{1}{c}{this paper} \\
	 \cline{1-4}
	 \bottomrule
 	\end{tabular}
 	\label{contributions}
 \end{center}
\end{table*}

{\color{blue}In summary, most of the current literature on RIS network analysis still relies on the binary system stochastic geometry, without fully considering the coupling relationships between the direct and reflected links. 
However, for multi-unit systems, it is crucial to incorporate the geometric relationships among different network units to accurately assess system performance.
To address this, we propose a ternary stochastic geometric theory for the RIS-assisted UDN system, utilizing a typical triangle to analyze the ternary network framework. 
Leveraging the typical triangle, we deduce and analyze a more comprehensive performance analysis of the RIS-assisted UDN system, as outlined in Table I.
The major contributions of this paper can be summarized as follows:
\begin{enumerate}
\item
We introduce three independent homogeneous Poisson point processes (HPPPs) to form a ternary point process, which effectively models and analyzes the RIS-assisted UDN system. 
To fully capitalize on the signal-enhancing capabilities of RIS, we optimize the RIS phase matrix to maximize the reflected signals using the maximum ratio acceptance principle.
Additionally, we conduct an analysis of the distribution of small-scale fading for the cascaded channels of the reflected signals from RISs within the system;
\item
We introduce the concept and analytical method of typical triangles by establishing a dual coordinate system using dual observation points.
Within this framework, we analyze the association relationship between BS and RIS, as well as BS and UE.
By examining the geometric relationship within the typical triangle, we further investigate the distribution of distances and angles. 
These distributions form the basis for modeling and analyzing the distribution of RIS reflection states, providing a more profound understanding of the node coupling within the RIS-assisted UDN system. 
\item
We introduce the ternary Campbell's theorem and ternary probability generating functional (PGFL) within the stochastic geometry framework to analyze the ternary network structure. 
Subsequently, we derive and analyze approximative closed formulae for signal/interference statistical characteristics, coverage probability, area spectral efficiency (ASE), and area energy efficiency (AEE) in the RIS-assisted UDN system.
Furthermore, we propose the concept of energy coverage efficiency (ECE) in the RIS-assisted UDN system and provide its derivation. 
This metric allows us to assess the efficiency of energy coverage within the system.
\end{enumerate}
}

The remaining sections of this paper are structured as follows.
Section II provides an overview of the system model, which encompasses the ternary stochastic geometry model, signal transmission model, and the design of the phase matrix of the RIS.
In Section III, we introduce a dual-coordinate analysis method and the concept of a typical triangle within the framework of ternary stochastic geometry. 
These tools are utilized for analyzing the RIS-assisted UDN system.
Section IV presents the theoretical analyses deduced from the aforementioned models and concepts.
Simulation results are presented in Section V, followed by concluding remarks drawn in Section VI.

\section{System model}
In this paper, $a$ and $\bf{A}$ respectively stand for a scalar and a matrix. 
Define the real number set and complex number set by $\mathbb{R}$ and $\mathbb{C}$, respectively.
$\mathcal{P}( \lambda  )$ represents a Poisson distribution with parameter $\lambda$, while $Nakagami (\varsigma,1)$ denotes a Nakagami distribution with parameter $\varsigma$. 
${\mathcal U}( a,b )$ indicates a uniform distribution within the range $[a,b]$.
$\mathcal{N}(\mu, \sigma^2)$ and $\mathcal{CN}(\mu, \sigma^2)$ are Gaussian distribution and complex Gaussian distribution with  mean $\mu$ and  variance $\sigma^2$, respectively. 
Let 
$\rm{E}[ \cdot ]$ and $\rm{D}[ \cdot ]$ be the mean and variance of a random variable.
Define ${\text{diag} } [a_1, a_2, \cdots, a_n] $ by a diagonal matrix whose main diagonal elements are $a_1$, $a_2$, $\cdots$, $a_n$ in order.  
The superscripts of the matrix ${\bf A}^{\rm{T}}$ and ${\bf A}^{*}$ respectively represent the transpose and conjugate of a matrix ${\bf A}$.

\subsection{Ternary stochastic geometry model}

Assume that the distributions of BSs, RISs and UEs are homogeneous Poisson point processes (HPPPs) defined by $\Lambda_N$, $\Lambda_M$ and $\Lambda_U$, with intensities $\lambda_n$, $\lambda_m$ and $\lambda_u$, respectively.
According to the point generation process of HPPP \cite{Stochastic_geometry_and_its_applications}, the number of BSs, RISs and UEs in the given area are respectively defined by $N$, $M$ and $U$, which obey independent Poisson distributions with different parameters.
However, the locations of points obey uniform distribution in the given area, where the area can be modeled as a circle with radius $R$. 
For example, the number of points $N$ in $\Lambda_N$ obeys Poisson distribution $\mathcal{P}(\lambda_n \pi R^2)$ and the locations of points are generated in the polar coordinate system by a ``radial generation" method, i.e., the probability density function (PDF) of polar diameter is $ f ( x ) = {2 x /}{R^2}$ $(0 \le x \le R)$, and the polar angle $\theta$ obeys uniform distribution on $[0, 2 \pi)$.
Assume the placement angle of RIS $\kappa_m$ obeys uniform distribution on $[0, 2\pi)$.

In stochastic geometry, the voronoi diagram is frequently used for describing cell distribution.
Through the vertical intersection of two nearby points, it divides the region into several voronoi cells.
There are three types of nodes for the RIS-assisted UDN: BS nodes, RIS nodes, and UE nodes. As a result, three different voronoi diagram types are taken into account.
An illustration of the voronoi diagram for a RIS-assisted UDN system is shown in Fig. \ref{voronoi}, where $N=15$, $M=20$, and $U=20$.
It is found that the voronoi diagrams of the BS, RIS and UE may exhibit overlapping regions, which further complicates the analysis of multiple voronoi diagrams.
\begin{figure}[t]
    \centering
	\includegraphics[width=7cm,height=6cm]{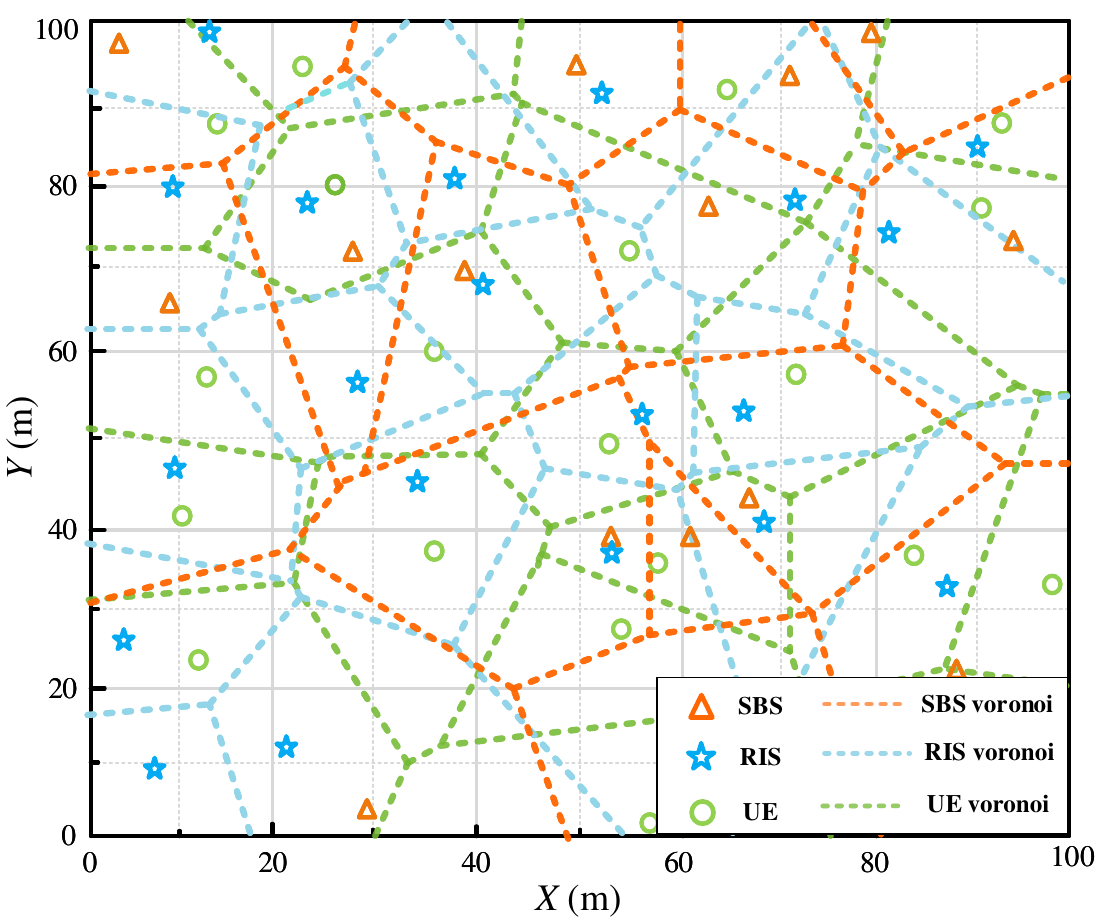}	
	\caption{The voronoi diagram of a RIS-assisted UDN, where $N=15$, $M=20$ and $U=20$.}
	\label{voronoi}
\end{figure}

{\color{blue} 
To facilitate a more comprehensive analysis, it is essential to examine the association rules among the RIS, the UE, and the BS.
In general, a UE tends to connect to the BS that offers the highest signal power.
Consequently, when all the BSs have the same transmit power, the UE may select the nearest BS as its serving BS.
In other words, the UE is within the voronoi cell of its serving BS \cite{SG_multi-tier}.
As $\Lambda_N$ and $\Lambda_U$ are independent of each other, there may exist the scenario, where no UE is inside the voronoi cell of a BS. 
In this case, the BS may choose to enter a dormant state to mitigate network interference.
The RIS employs a similar strategy by selecting the nearest BS to provide its services.
This operation enables the efficient transmission of control signaling between the BS and the RIS through the wired link. 
Equivalently, this framework also aligns perfectly with the concept of advocating the BS to exercise control over the RISs within its voronoi cell.
This arrangement can enable the BS to enhance the UE performance by managing and optimizing the operations of the RISs.
}

\subsection{Signal transmission model}

This paper considers the narrow-band communication downlink scenario, which ignores the effect of transmission delay. 
{\color{blue} To facilitate the analysis of the topology performance, it is assumed that each equipment terminal (e.g., BS, and UE) holds one antenna, and each RIS has $Q$ reflection elements.}
In addition, the RIS is assumed to reflect electromagnetic waves by only one side with the tunable phase, without considering the tunable amplitude and refraction function (transmitting-RIS) \cite{coverage_probability_ris_assisted,SG_RIS_CELLULAR,wireless_network_RIS,Large_Scale_Deployment_of_Intelligent_Surfaces}. 
{\color{blue} Without loss of generality, it is assumed that one BS, one RIS and one UE constitute the basic unit in the ternary system, for comparison with the binary system of the classical UDN \cite{coverage_probability_ris_assisted,SG_RIS_CELLULAR,Spatial_Throughput_Characterization,Hybrid_Active/Passive_Wireless}. }
Taking the $n$th BS as an example, it is assumed that the $u$th UE is served by the $n$th BS, and the set of RISs serving the $n$th BS is defined by $\Psi_n^M$.
Let the modulated signals of the $u$th UE be $x_{u}$, where $ {\rm E}\left[|x_{u}|^2\right] = 1$. 
Then, the transmit signal of the $n$th BS is expressed as $s_n =   \sqrt{P_{tr}}  x_{u}$,
where $ P_{tr}$ is the transmit power of the $n$th BS. 
In this paper, we assume all BSs have the same transmit power. 

The channel impulse response between the $n$th BS and the $u$th UE is denoted as $h_{n,u}$, the channel impulse response between the $n$th BS and the $m$th RIS is set to be ${\bf{w}}_{n,m}\!  =\!  [w_{n,m,1}, w_{n,m,2}, \cdots, w_{n,m,Q}]^T \!  \in\!   \mathbb{C}^{Q\times 1}$, and the channel impulse response between the $m$th RIS and the $u$th UE is ${\bf{g}}_{m,u} \! =\!  [g_{m,u,1}, g_{m,u,2}, \cdots, g_{m,u,Q}]$ $\in  \mathbb{C}^{1 \times Q}$. 
Moreover, the reflection matrix of the $m$th RIS is defined by ${\bf{\Phi}}_m \triangleq \text{diag} \left[ {\phi_{m,1}}, {\phi_{m,2}}, \cdots, {\phi_{m,Q}} \right]$, where $\phi_{m,q}$ is the controllable phase shift of the $q$th reflection element.
The distance between the $n$th BS and the $u$th UE is defined by $d_{n,u}^{D}$, the distance between the $n$th BS and the $m$th RIS is defined by $d_{n,m}^{I}$, and the distance between the $m$th RIS and the $u$th UE is $d_{m,u}^{R}$, where the subscripts ``$D$", ``$I$", and ``$R$" indicate ``direct", ``incident" and ``reflect", respectively. 
\begin{figure*}[t]
	\centering
	\includegraphics[height=4.9cm]{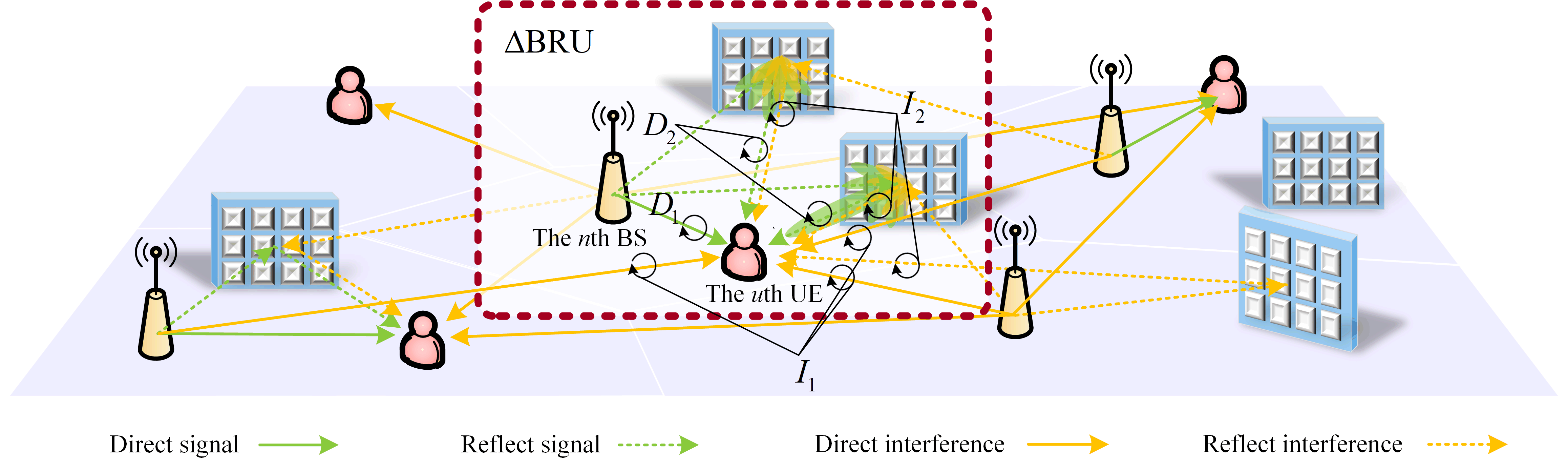}
	\caption{Signal transmission diagram of RIS-assisted UDN, where $D_1$,$ D_2$, $I_1$ and $I_2$ respectively stand for the direct signal, the reflected signal, the direct interference and the reflection interference during the communication between the $u$th UE and the $n$th BS.}
	\label{system}
\end{figure*}

{\color{blue} In this paper, the small-scale fading channels are modeled as the Nakagami distribution $Nakagami(\varsigma,1)$, and the large-scale fading channel only considers the path loss, i.e., $\textstyle (1+d)^{-\alpha}$, where $\alpha$ is the path loss coefficient and $d$ indicates the signal transmission distance.}
Therefore, the received signal of the $u$th UE is
	\begin{equation}
		{\small
			\begin{aligned}
				\label{y_u}
			   y_u \!= \!&\!\! \sum_{n=1}^{N} \! {{ \!\frac{h_{n,u} s_n}{(1 \!+ \!d_{n,u}^{D} \!)^{\frac{\alpha}{2} }} \!}}  \!
					+ \! \! \!\!\sum_{m=1}^{M} \! \! \frac{ {\bf{g}}_{m,u} {\bf{\Phi}}_m }{(1 \!+ \!d_{m,u}^{R} \!)^{\frac{\alpha}{2} } }\! \! \!
					   \!\left( \sum_{n=1}^{N} \!{ { \!\frac{{\bf{w}}_{n,m}  \beta_{n,m, u}  s_n}{(1 \!+ \!d_{n,m}^{I} \!)^{ \frac{\alpha}{2} }}\!}} \!\cdot \! \!\right) \!\!+ \!\!z_u\\
					   \!=  & \underbrace{ {{ \!\frac{h_{n,u}  s_n}{(1 \!\!+ \!d_{n,u}^{D}\! )^{\frac{\alpha}{2}}\!} }} }_{D_1} 
					    + \!\!  \underbrace{\sum_{m=1}^M    \! \frac{ {\bf{g}}_{m,u} {\bf \Phi}_m {\bf w}_{n,m}   \beta_{n,m, u}    s_n}{ (1 \!+\! d_{m,u}^{R} \! )^{ \frac{\alpha}{2} }\!  (1 \!+ \!d_{n,m}^{I} \! )^{\frac{\alpha}{2} }} \!}_{D_2} \!\\
					 & + \!\!\!\! \underbrace{\sum_{ \substack{n'\! \neq n,\\ n'\!=1}}^{N}  \!\!\!\!\frac{ h_{n'\!,u}  s_{n'\!}}{ (1 \!\!+ \!d_{n'\!,u}^{D} \! )^{ \frac{\alpha}{2} }}\!}_{I_1} 
					+  \!\! \!\underbrace{\sum_{m=1}^{M}\!  \sum_{\substack{n'\! \neq n,\\ n'\!=1}}^{N} \!\!\!\frac{{\bf{g}}_{m , u}\! {\bf \Phi}_m {\bf w}_{n'\!,m} \beta_{n'\!,m, u}  s_{n'\!}}{ (1 \!+ \!d_{m,u}^{R} \! )^{ \frac{\alpha}{2} } (1 \!+ \!d_{n'\!,m}^{I}\!  )^{  \frac{\alpha}{2} } }
					    \!}_{I_2}    +  \!z_u,
				\end{aligned}}
	\end{equation}%
where $z_u \sim \mathcal{N}(0,\sigma_n^2)$ represents the additive white Gaussian noise (AWGN) of the $u$th UE.
The parameter $\beta_{n,m, u} \! \in\!  \{0,1\}$ denotes the reflection state of the reflected signals from the $n$th BS through the $m$th RIS to the $u$th UE.
If $\beta_{n,m, u} \! \!=\!\!  1$, it indicates that the $m$th RIS can effectively reflect the signals from the $n$th BS to the $u$th UE. In other words, the RIS is capable of assisting in the transmission for the $u$th UE. On the contrary, if $\beta_{n,m, u} \! \!=\!\!  0$, it suggests an interrupted link, implying that the RIS does not contribute to the signal transmission for the $u$th UE.
The first term in equation (\ref{y_u}), referred to as $D_1$, represents the direct signal from the $n$th BS to the $u$th UE. The second term, denoted as $D_2$, accounts for the reflected signals from the RISs received by the $u$th UE. Additionally, the third and fourth terms respectively represent the interference originating from other BSs and RISs, which are defined as $I_1$ and $I_2$, as illustrated in Fig. \ref{system}.

For further discussion, the cascade channel impulse response from the $n$th BS through the $m$th RIS to the $u$th UE is defined by $h_{n,m,u} = {\bf {g}}_{m,u} {\bf \Phi}_m {\bf{w}}_{n,m} $.
Consequently, the signal to interference plus noise ratio (SINR) of the $u$th UE $\gamma_u$ is given by 
\begin{equation}
	{\small
	\label{SINR}
		\begin{aligned}
			\gamma_u \! 
			    = \!\!\frac{ P_{tr}\left|  \frac{ h_{n,u} }{(1 + d_{n,u}^{D}  )^{\frac{\alpha}{2} } }    + \sum\limits_{m=1}^M 
			   \frac{  h_{n,m,u}    \beta_{n,m, u}}{ \left(1 + d_{m,u}^{R} \right)^{ \frac{\alpha}{2}}  \left(1 + d_{n,m}^{I} \right)^{ \frac{\alpha}{2}} }  \right|^2}
	   { P_{tr}\!\!\sum\limits_{\substack{n'\! \neq n, \\  n'\!=1}}^{N}   \!\!\frac{|h_{n'\!,u}|^2}{ (1 \!+ \!d_{n'\!,u}^{D} \! )^{ \alpha }} \! +\! P_{tr}\!\!\sum\limits_{m=1}^{M}  
				    \!\sum\limits_{\substack{n'\! \neq n,\\ n'\!=1}}^{N}  \!\! \frac{|h_{n'\!,m,u}|^2    \beta_{n'\!,m,u}^2 } { \left(1 \!+ \!d_{m,u}^{R}\right)^{ \alpha }(1\! + \!d_{n'\!,m}^{I}\!  ) ^{ \alpha }}\!   \! \!+ \! \sigma_n^2 }
		\end{aligned},}
\end{equation}%
where the cascade channel impulse response $h_{n,m,u} $ has a significant impact on $D_2$ and $I_2$, especially for the gain of combining signals of the UE.
Therefore, we can design the phase matrix  ${\bf{\Phi}}_m$ of RIS, to directly control the cascade channel to the maximum modulus state.

\subsection{Design the phase matrix of RIS}
{\color{blue} In general, it is difficult to obtain channel state information of other BSs, especially for RIS-assisted systems, which have too much pilot overhead and too high latency. 
}
Since this paper only considers the phase tuning capability of RIS, the phase matrixes of RISs are designed to maximize the power of signals.
Thus, the controllable phase shift of the $q$th reflection element of the $m$th RIS is 
	\begin{equation}
		{\small
\label{phi_q}
	\phi_{m,q} \!= \!\! 
		 \frac{ {h_{n,u}}}{| h_{n,u}|} \!\cdot \!\frac{ {w^*_{n,m,q}}}{| w_{n,m,q}|} \!\cdot\! \frac{{g^*_{m,u,q}}}{|g_{m,u,q}|}, ~1 \!\le\! q \!\le\! Q,~m \in \Psi_{n}^{M},}
\end{equation}%
where $\phi_{m,q}$ is determined by the product of three parts, i.e., the phase of $h_{n,u}$, $w^*_{n,m,q}$ and $g^*_{m,u,q}$. 

Consequently, the cascade channel gain from the $n$th BS through the $m$th RIS to the $u$th UE is expressed as
\begin{equation}
	{\small
		\left|h_{n,m,u} \right| =  \left\{
	\begin{aligned}
		&\sum_{q=1}^Q   |g_{m,u,q}|  \cdot |w_{n,m,q}|,~ m \in \Psi_{n}^{M},\\
		&\left|\sum_{q=1}^Q   g_{m,u,q} \cdot w_{n,m,q} \right|, ~~m \notin \Psi_{n}^{M} ,
	\end{aligned}\right. }
\end{equation}%
where $m \notin \Psi_{n}^{M}$ indicates that the $m$th RIS does not serve the $n$th BS.
{\color{blue}If $m \in \Psi_{n}^{M}$, the phase matrix ${\bf \Phi}_m$ is matched with the channel impulse response ${\bf w}_{n,m}$ and ${\bf g}_{m,u}$, and the reflected signals can be aligned. 
According to the central limit theorem, the PDF of $ |h_{n,m,u} | $ is approximately  a Gaussian distribution, i.e.,  $\textstyle \mathcal{N}\left( \frac{Q \Gamma^2( \varsigma+\frac{1}{2})}{\varsigma \Gamma^2(\varsigma)}, Q(1-\frac{\Gamma^4(\varsigma +\frac{1}{2})}{\varsigma^2\Gamma^4(\varsigma)})  \right)$. 
Furthermore, the phase of the cascaded channel $h_{n,m,u}$ should align with the phase of its direct link $h_{n,u}$, and thus the cascaded channel $h_{n,m,u}$ follows a uniform distribution $\mathcal{U} (0, 2\pi)$ \cite{nakagami}.
In contrast, if $m \notin \Psi_{n}^{M}$, the reflected signals become random variables to the $u$th UE, so $ h_{n,m,u}$ is approximately a complex Gaussian distribution  $\textstyle \mathcal{CN}\left( 0 , \frac{Q}{2} \right)$. 
}


\section{Ternary stochastic geometric modeling for RIS-assisted UDN}

A typical point (BS or UE) is generally {\color{blue} specified in} the classical stochastic geometry theory for evaluating the performance of the entire network, since the Slivnyak Theorem has proved that the statistical characteristics of HPPP are independent of the location of the observation point \cite{Stochastic_Geometry_for_Wireless_Networks}.
Due to the independence of its communication links, the traditional point-to-point communication only concerns the relative distance between the transmitter and receiver. 
Nevertheless, the RIS-assisted UDN system is a ternary system, which consists of three points, namely one BS, one RIS, and one UE, forming a triangle $\triangle {\rm BRU}$ geometric relationship, as shown in Fig. \ref{ternary_introduction}.
To represent the geometry relationships of the three points, this section proposes the typical triangle as the ternary stochastic geometric basis.

\begin{figure}[t]
    \centering
	\includegraphics[height=5cm]{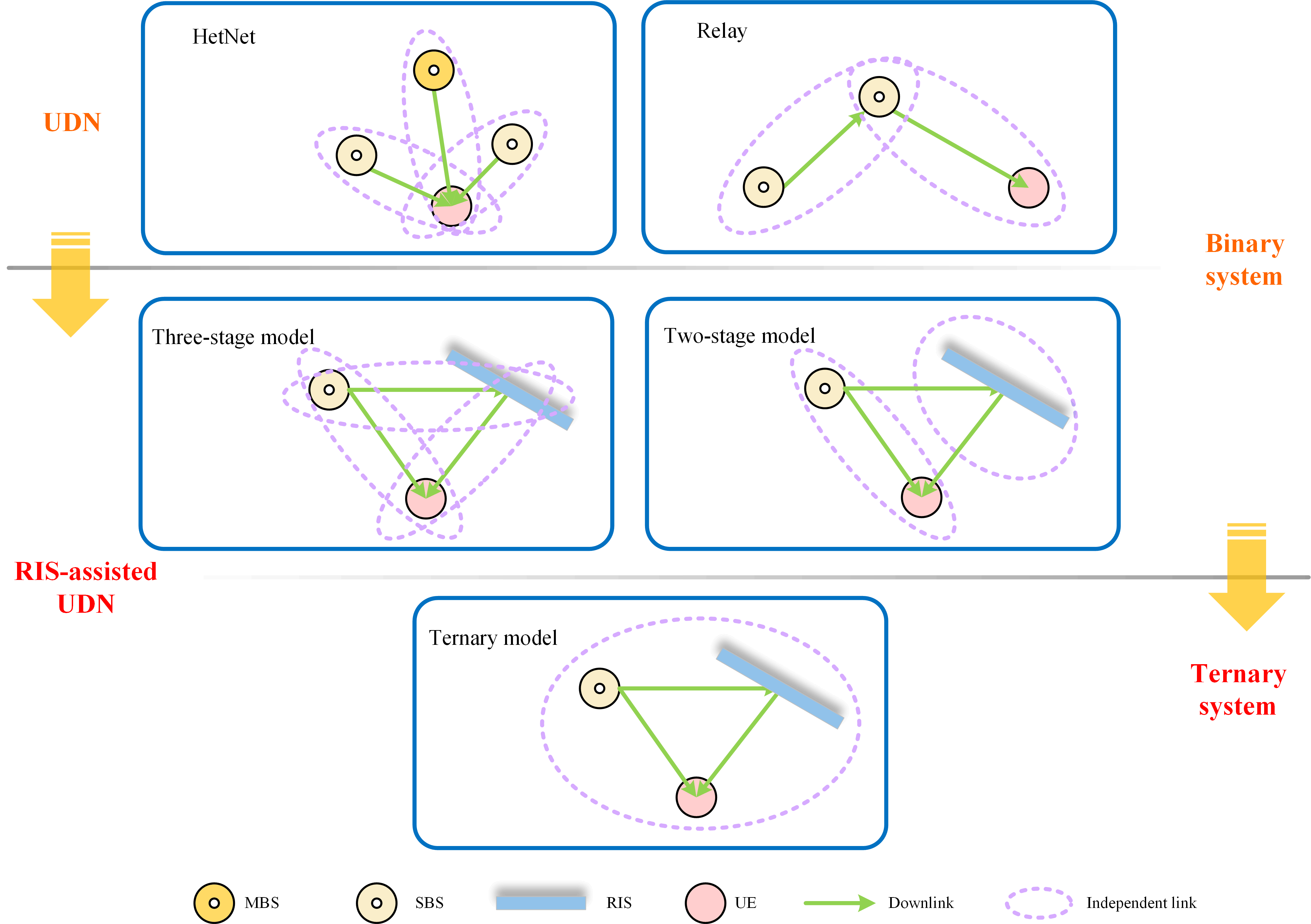}	
	\caption{Ternary systems differ from binary systems in previous literature in terms of unit modeling.}
	\label{ternary_introduction}
\end{figure}

{\color{blue}

\subsection{Dual-coordinate analysis method}


\begin{figure}[t]
    \centering
	\includegraphics[width=7cm]{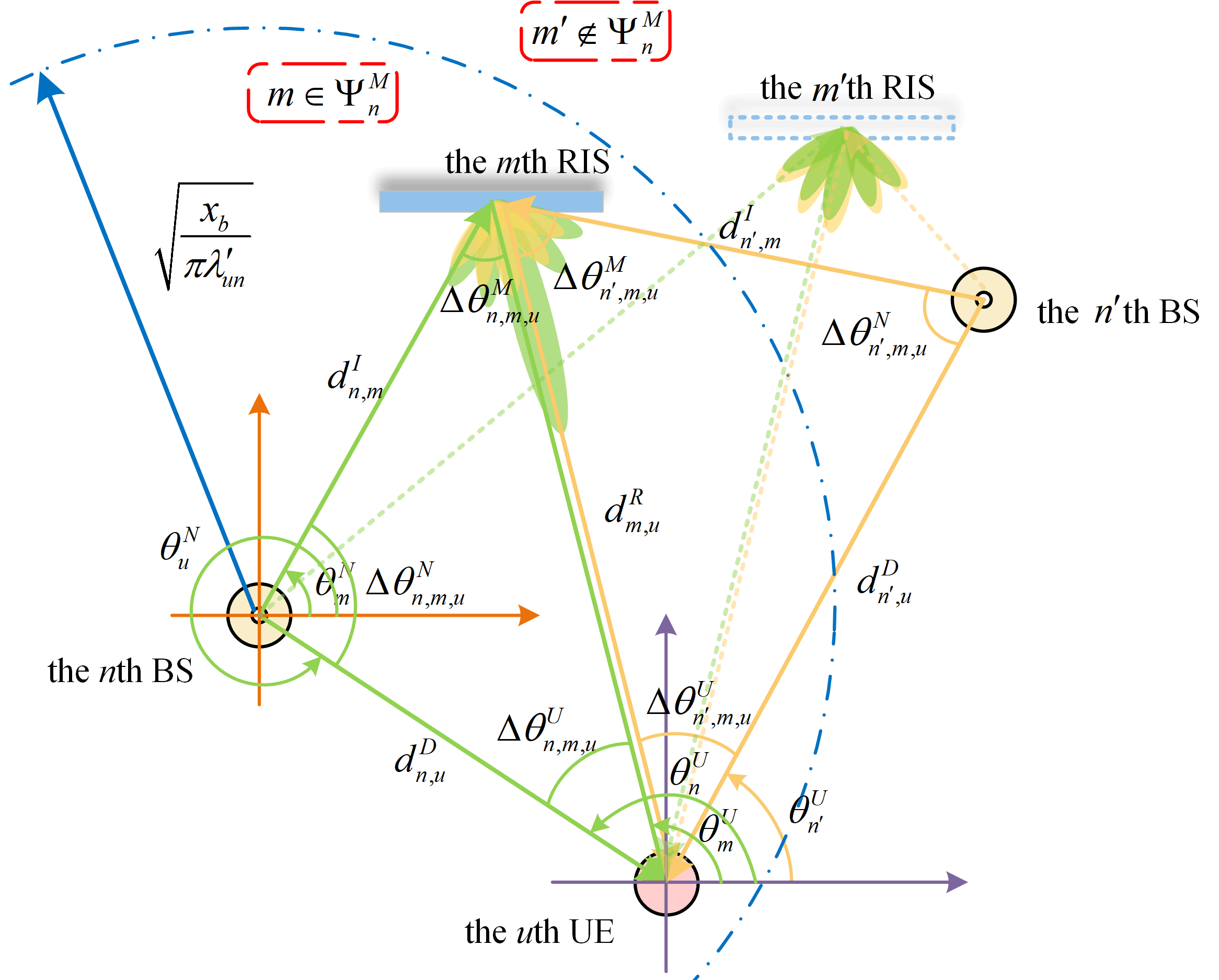}
	\caption{The geometric relationship of the $\triangle {\rm BRU}$, where the $n$th BS acts as the central point of a circle with a cell radius of $\sqrt{{x_b}/({\pi\lambda'_{un}})}$, effectively partitioning the RIS into two distinct regions -  $m\in \Psi_m^M$ (the interior of the cell) and $m \notin \Psi_m^M$ (the exterior of the cell).}
	\label{Geometrical_Relationship}
\end{figure}

As mentioned earlier, analyzing the RIS-assisted UDN system involves studying the geometric relationship of the $\triangle {\rm BRU}$.
Hence, it is difficult to determine the positional relationship of all points in the 3HPPP using a single typical point.
We need at least two points: one is the typical UE and the other is the typical BS, selected based on their connectivity relationship. 
As show in Fig. \ref{Geometrical_Relationship}, each typical point can establish a coordinate system to locate the positions of all points in the 3HPPP, thereby employing a dual coordinate system analysis method. 
To demonstrate the validity of the dual-coordinate analysis method, we present Proposition \ref{pro1}.

\begin{proposition}[Coexistence of observation points]\label{pro1}
    In the RIS-UDN system, when the association relationships among HPPPs are fixed, the PDFs of the correlation distances and angles remain consistent, regardless of the observation point from which they are viewed.
\end{proposition}

\begin{proof}
	Since HPPP has stationarity and isotropy \cite{Stochastic_geometry_and_its_applications,Stochastic_Geometry_for_Wireless_Networks}, it can be obtained that
		\begin{equation}
			{\small			{\rm P}[ \Lambda(S)\!=\!n]  = {\rm P}[ \Lambda(S_{+(x_b,y_b)})\!=\!n]  = {\rm P}[ \Lambda(S_{\curvearrowright \theta_m})\!=\!n] ,}
		\end{equation}%
	where ${\Lambda}(B)$ is the number of points of $\Lambda$ in the set $B$. 
	Similarly to the notations in \cite{Stochastic_geometry_and_its_applications}, $S_{+(x_b,y_b)} \triangleq \{(x,y)+(x_b,y_b), (x,y) \in S\}  $ and $S_{\curvearrowright \theta_m} \triangleq \{ (x,y)\cdot e^{-j\theta_m},(x,y) \in S \}$ denote the translation of the set $S$ by $(x_b, y_b)$ and the clockwise rotation $\theta_m$, respectively. 

	Moreover, in accordance with the Slivnyak Theorem, the typical point does not affect the distribution of the HPPP.
	Therefore, the coexistence of typical points can be proven by
		\begin{equation}
			{\small
			\begin{aligned}
				{\rm P}[ S \in Y || (0,0)] = & {\rm P}[ S_{+(x_b,y_b)} \in Y || (0,0)] \\
				= & {\rm P}[ S \in Y || (0,0), (x_b,y_b)] \\
				= & {\rm P}[ S \cup \{ (0,0)\} \cup  \{ (x_b,y_b)\} \in Y ],
			\end{aligned}}
		\end{equation}%
	where ${\rm P}[S\in Y||(x,y)]$ denotes the probability that $S$ has property $Y$ from the observation point $(x,y)$.

\end{proof}

In the following analysis, the geometric relationship and association relationship of each point in the 3HPPPs will be examined using two coordinate systems established by a typical UE and a typical BS. 
This method allows for a comprehensive understanding of the system's geometry and association characteristics.

\subsubsection{Analysis of association relationship}

To facilitate the subsequent analysis, we first define a typical $\triangle {\rm BRU}$, as depicted in Definition \ref{define}.
\begin{Define}\label{define}
	A triangle, denoted as $\triangle {\rm BRU}$, is consisted of a BS, a RIS, and a UE. 
	Additionally, a triangle with at least one typical point (BS or UE) is defined as a typical $\triangle {\rm BRU}$.
\end{Define}

In this paper, we initially establish a coordinate system using the typical UE as the origin. 
This facilitates the localization of other BSs for the analysis of direct signals and interference. 
The typical UE selects the nearest BS to serve it, which is regarded as the typical BS. 
Subsequently, we use the typical BS as the origin to conveniently locate the RIS for the analysis of reflected signals and interference.
Specific symbolic notations and definitions in the $\triangle {\rm BRU}$ can be shown in Table \ref{notations} and Fig. \ref{Geometrical_Relationship}.
}

\begin{table}[t]
	{\color{blue}
	\begin{center}
		  \caption{List of notations in $\triangle {\rm BRU}$}
	\begin{tabular}{|c |c|}
	   \toprule
	   \specialrule{-0.05em}{1pt}{1pt}
		\cline{1-2}
   \cline{1-2}
	 \multicolumn{1}{c}{ Notation }                   & \multicolumn{1}{c}{Description}   \\ 
	\cline{1-2}
	\multicolumn{1}{c}{$\Delta \theta^i_{n,m,u}$}                   & \multicolumn{1}{c}{ the angle in  $\triangle _{n,m,u}$ with point $i$ as vertex}         \\ 
	\multicolumn{1}{c}{$\Psi_n^M$}                   & \multicolumn{1}{c}{ the set of RIS-assisted the $n$th BS}         \\ 
	\multicolumn{1}{c}{  \multirow{1}{*}{$\theta^I_j$}  }                    & \multicolumn{1}{c}{ polar angle: the $j$th equipment when $I$ is at the origin\footnotemark[1]}         \\ 
	\multicolumn{1}{c}{$\kappa_m$}                   & \multicolumn{1}{c}{ the placement angle of the $m$th RIS}         \\ 
	\multicolumn{1}{c}{$d^D_{n,u}$}                   & \multicolumn{1}{c}{ distance: the $n$th BS and the $u$th UE}         \\ 
	\multicolumn{1}{c}{$d^I_{n,m}$}                   & \multicolumn{1}{c}{ distance: the $n$th BS and the $m$th RIS}         \\ 
	\multicolumn{1}{c}{$d^R_{m,u}$}                   & \multicolumn{1}{c}{ distance: the $m$th RIS and the $u$th UE}         \\ 
	\bottomrule
	\end{tabular}
	\label{notations}
\end{center}
\footnotemark[1]{$I \in \{N,U\}$}, where $N$ and $U$ denote a typical BS and a typical UE, respectively.
}
\end{table}

The PDF of the area of voronoi cell is given as 
	\begin{equation}
		{\small
		\label{size_voronoi}
		 f_{X_{\mathcal{B}}}(x_b)=\frac{3.5^{3.5}}{\Gamma(3.5)} x^{2.5}_b e^{-3.5x_b},}
	\end{equation}
where $x_b$ is the random variable that reflects the normalized voronoi cell area of BS \cite{size_distribution_poisson_Voronoi_cells}.
According to \eqref{size_voronoi}, the probability of at least one UE located in the BS voronoi cell is 
	\begin{equation}
		{\small
			\begin{aligned}
				\label{uge0}	
			{\rm{P}} \left[  u \neq 0  \right] = 1-(1 + \frac{\lambda_u}{3.5 \lambda_n})^{-3.5}
			,
			\end{aligned}}
	\end{equation}
where ${\rm{P}} \left[  u \neq 0  \right]$ is the probability of an active BS \cite{Base_Station_Density_Voronoi_cells}.
According to the thining theorem of HPPP \cite{Stochastic_Geometry_for_Wireless_Networks,Stochastic_geometry_and_its_applications}, the intensity of active BS is $\lambda'_{un} \triangleq {\rm{P}} \left[  u \neq 0  \right] \lambda_n$. 
{\color{blue} 
This is a critical conclusion, as it allows for the analysis of system performance in two HPPPs $\Lambda_N$ and $\Lambda_U$ by using a typical UE and an active BS HPPP $\Lambda'_{UN}$.
}

{\color{blue} So far, the analysis based on the typical UE remains consistent with the original binary stochastic geometry. 
For the RIS-assisted UDN system, the RISs can be categorized into two types based on their distance relationship with the reference BS: those inside the reference BS voronoi cell ($m \notin \Psi^M_n$) and those outside the cell ($m \notin \Psi^M_n$).
RISs inside the voronoi cell have the capability to perform beamforming, while those outside the cell cannot align with the signals.
This division results in the Poisson point process  $\Lambda_M$ being split into two parts.
Furthermore, if the $m$th RIS is within the voronoi cell of the $n$th BS, it means that the distance from the $m$th RIS to the $n$th BS is less than the distance to the other BSs, i.e., $\Psi^M_n=\{ \forall d^I_{n,m} \le d^I_{n',m}, n' \in \Lambda_N/\{n\} \}$.
It is worth noting that we cannot refine the analysis by the thining RIS intensity process for the RIS-assisted UDN system, because of the associated distance $d^I_{n,m}$.
Therefore, to accurately characterize the association relationships of RISs, we introduce a coordinate system centered around a typical BS.

Based on the  Gilbert disk (Boolean) model described in \cite{Stochastic_Geometry_for_Wireless_Networks}, the probability of the $m$th RIS being located in the typical voronoi cell can be expressed as 
	\begin{equation}
		{\small
		\begin{aligned}
			 {\rm P}\Big[  d^I_{n,m} \!\le \!\!\sqrt{\frac{x_b}{\pi \lambda'_{un}}}  \Big|    d^I_{n,m} \Big]
			 \!\!=  & \!\!\!\int_{\!\pi\lambda'_{un} \!(d^I_{n,m})^2 }^{\infty} \!\!\frac{3.5^{3.5}}{\Gamma(3.5)}x_b^{2.5} e^{-3.5 x_b}  {\text d} x_b\\
			 \!= & 0.3\Gamma(3.5, 3.5 \pi\lambda'_{un} (d^I_{n,m})^2 )\\
			 \! \triangleq & p(\lambda'_{un}, d^I_{n,m}),
		\end{aligned}}
	\end{equation}%
where $\Gamma(a,x) \triangleq \int_x^{\infty} t^{a-1} e^{-t} {\text d} t $ is the incomplete gamma function. 
Although it is impossible to prove ${\rm P}[\Psi^M_n | d^I_{n,m} ] \Leftrightarrow {\rm P}[ d^I_{n,m} \le \sqrt{x_b/(\pi \lambda'_{un})} | d^I_{n,m} ]$ in the Gilbert disk analysis, numerical simulations have shown a perfect fit between the two. 
This observation further supports the validity and reasonableness of the analysis conducted.
}

\subsubsection{Analysis of geometric relationship in typical $\triangle {\rm BRU}$}
To unify the analysis, the typical UE is defined as the first UE and the typical BS is defined as the first BS.
From the perspective of the typical UE (located at the origin), the PDF of $d_{1,1}^{D}$  is derived as follows
	\begin{equation}
		{\small
		f\left({d_{1,1}^{D}}\right)= 2 \pi \lambda'_{un} {d_{1,1}^{D}} \cdot e^{-\pi \lambda'_{un} \left({d_{1,1}^{D}}\right)^2}.
		\label{dn_11}}
	\end{equation}

In addition, we can also find that the typical UE’s vertex angle $\textstyle \Delta \theta_{1,m,1}^{U}=\theta^{U}_m-\theta_{u_1}^{U}$  obeys the uniform distribution in the range of $[0, 2\pi)$, since the polar angle $\theta^{N_M}_m$ and $\theta_{1}^{N_U}$ are independent identically distributed (i.i.d.) and both obey the uniform distribution on $[0, 2\pi)$.
According to the law of cosines, it can be obtained that
	\begin{equation*}
		{\small
		\textstyle d_{1,m}^{I} = \sqrt{(d_{1,1}^{D})^2 + (d_{m,1}^{R})^2 - 2d_{1,1}^{D} d_{m,1}^{R} \text{cos}( \Delta \theta_{1,m,1}^{U})},}
	\end{equation*}%
where the PDF of $d^{R}_{m,1}$ also obeys the polar diameter $f(d^{R}_{m,1})=2 d^{R}_{m,1} /R^2, (0\le  d^{R}_{m,1}  \le R)$.


\subsection{Mathematic model of the reflection state of a RIS}\label{section2_1}

The $\triangle {\rm BRU}$ consists of three transmission links (or called sides): from BS to RIS $(B \to R)$, from BS to UE  $(B \to U)$, and from RIS to UE $(R \to U)$.
The interrupted link is generated by the placement angle $\kappa_m$ of the RIS, because of the one-side reflection.
Let $\beta_{n,m, u} \in \{0,1\}$ be the reflection state of the RIS. 
If $\beta_{n,m, u} = 1$, the $m$th RIS can reflect the signals from the $n$th BS to the $u$th UE.
Conversely, if the link is interrupted, it indicates that the RIS is not functioning, i.e., $\beta_{n,m, u} = 0$.
Thus, there are generally four cases: full connection, the interrupted link of ${B \nrightarrow R}$, the interrupted link of ${R \nrightarrow U}$, and the interrupted link of ${B \nrightarrow R \nrightarrow U}$, as depicted in Fig. \ref{RIS_model}.
\begin{figure}[t]
	\centering
	\subfigure[~$\beta_{n,m, u} = 1$]
	{
		\includegraphics[height=1.9cm,width=3.7cm]{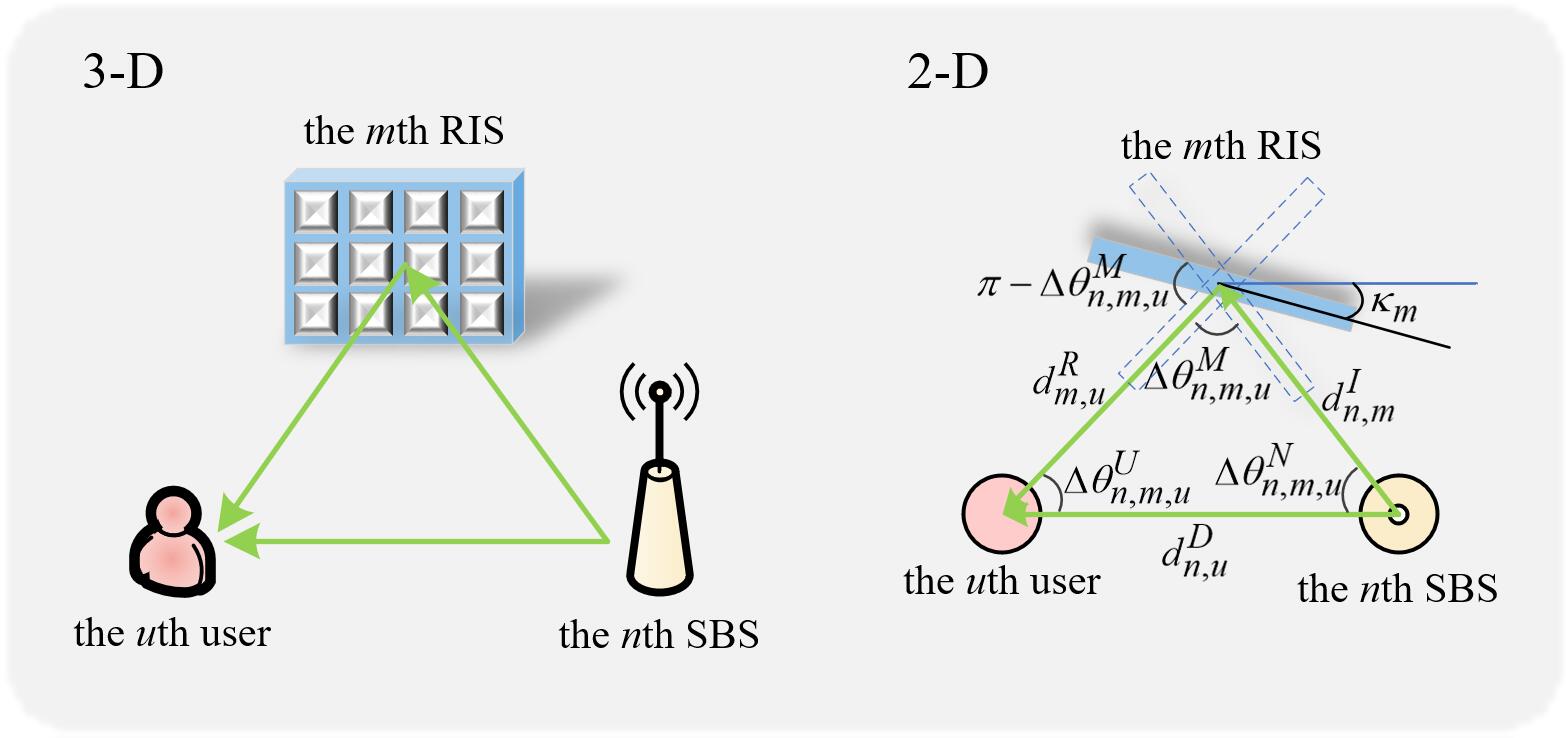}
	}
	\quad
	\subfigure[~$\beta_{n,m, u} = 0$ and ${B \nrightarrow R}$]
	{
		\includegraphics[height=1.9cm,width=3.7cm]{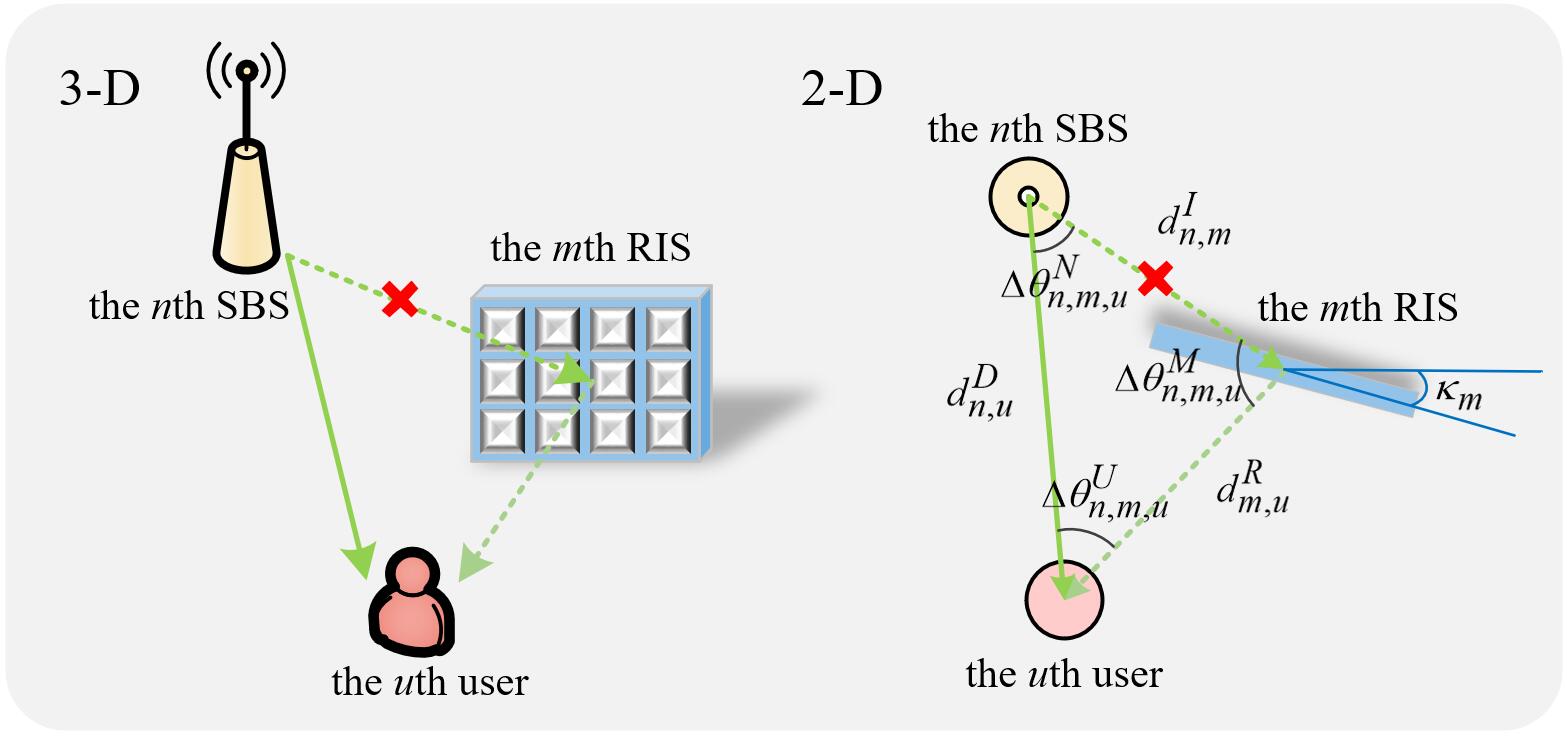}
	}
	\quad
	\subfigure[~$\beta_{n,m, u} = 0$ and ${R \nrightarrow U}$]
	{
		\includegraphics[height=1.9cm,width=3.7cm]{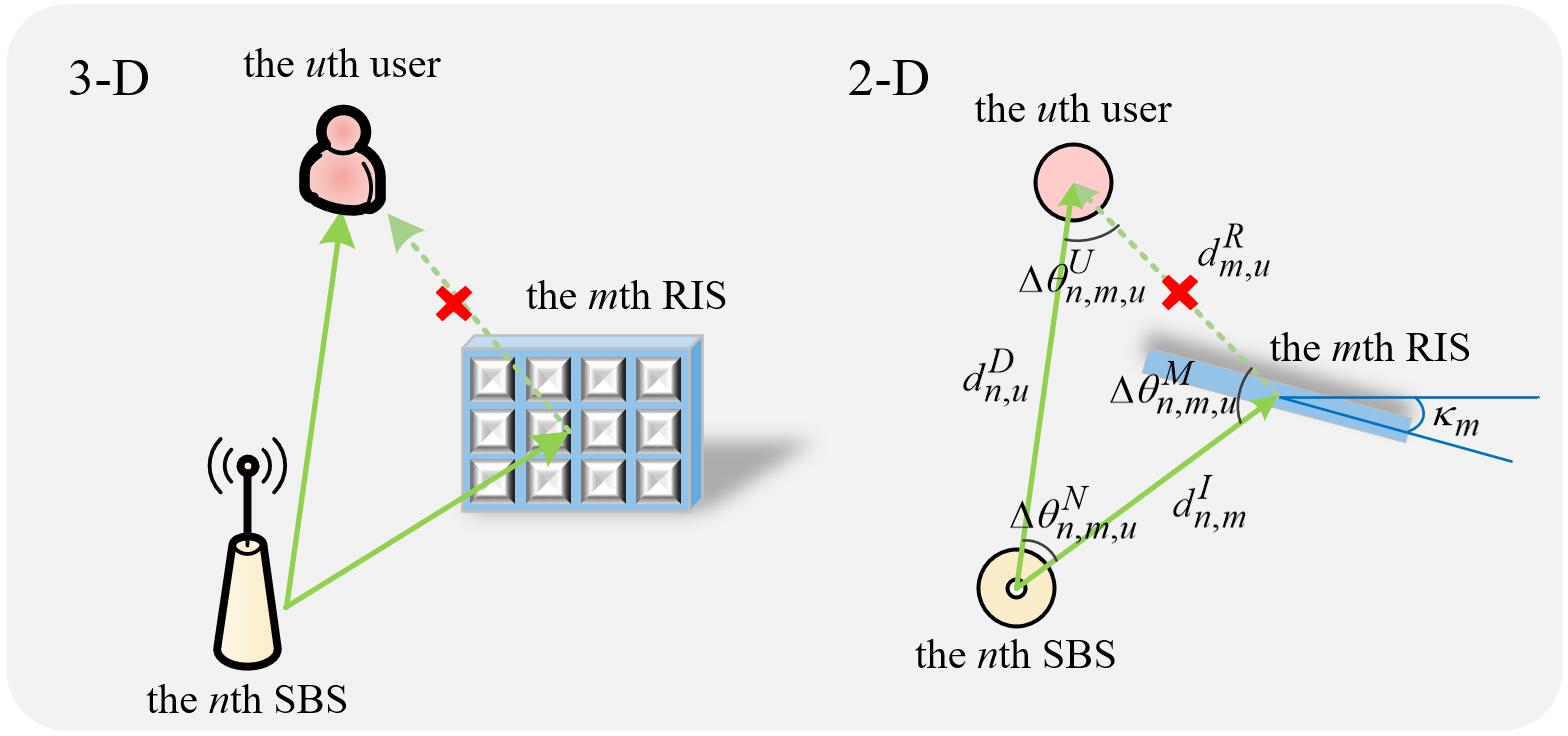}
	}
	\quad
	\subfigure[~$\beta_{n,m, u} = 0$, ${B \nrightarrow R}$ and ${R \nrightarrow U}$]
	{
		\includegraphics[height=1.9cm,width=3.7cm]{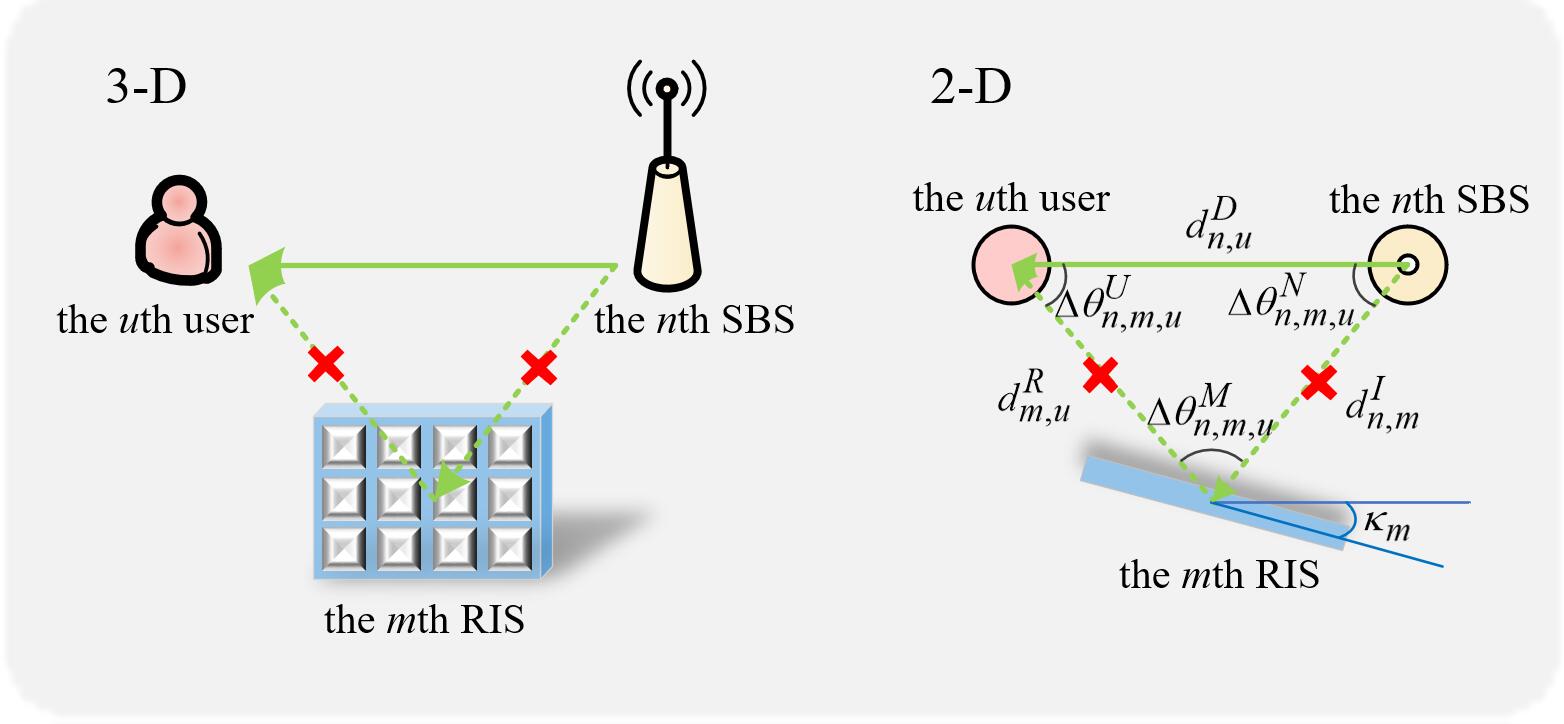}
	}
	\caption{Diagram of four reflection cases in the network, where  ``solid line" represents the transmission link and ``dashed line" represents the interrupted link. (a) full connection; (b) the interrupted link, i.e., ${B \nrightarrow R}$; (c) the interrupted link, i.e., ${R \nrightarrow U}$; (d) the interrupted link, i.e., ${B \nrightarrow R \nrightarrow U}$.}
	\label{RIS_model}
\end{figure}

According to the geometric relationship of $\triangle {\rm BRU}$, the successful reflection placement angle range is from $- \theta^{N_M}_m$ to $\pi - \Delta \theta^M_{n,m,u} -  \theta^{N_M}_m$, as shown in Fig. \ref{RIS_model}.
Consequently, the reflection state $\beta_{n,m, u}$ of the $m$th RIS is expressed as
	\begin{equation}
		{\small
		\beta_{n,m, u} = \left\{
		\begin{aligned}
		1,~~~~~ &    - \theta^{N_M}_m \le  \kappa_m \le  \pi - \Delta \theta^M_{n,m,u} -  \theta^{N_M}_m, \\
		0,~~~~~ & ~~~~~~~~~~~~~~  \text{others},
		\end{aligned}
		\right.  }
		\end{equation}%
where $ \theta^{N_M}_m$ is the polar angle of the $m$th RIS and $\Delta \theta^M_{n,m,u}$ is the $m$th RIS's vertex angle.

For a given angle $\Delta \theta^M_{n,m,u}$, the conditional probability of $\beta_{n,m, u}=1$ is calculated as
	\begin{equation}
		{\small
		\begin{aligned}
			\label{betaeq1delta} 
			&{\rm{P}}\left[ \beta_{n,m, u}=1 \Big| \Delta \theta^M_{n,m,u} \right]
		=   \frac{\pi - \Delta \theta^M_{n,m,u}}{2\pi} \le \frac{1}{2},
		\end{aligned}}
	\end{equation}%
where $\Delta \theta^M_{n,m,u} \!\!=\!\! {\rm{arccos}} \Big[  \frac{ ({d_{m,u}^{R}})^2 + ( {d_{n,m}^{I}} )^2 - ({d_{n,u}^{D}})^2}{2 {d_{m,u}^{R}} {d_{n,m}^{I}} }    \Big]$.
Considering the triangular geometric relationships, it becomes apparent that the RIS reflection state is less than $1/2$. 
This finding provides a slightly more accurate assessment compared to previous literature that assumed a half RIS operation \cite{SG_RIS_CELLULAR,coverage_probability_ris_assisted,RIS_spatially,interference_RIS_intelligent}.

\section{Ternary Stochastic geometric analysis for RIS-assisted UDN}
In the previous section, we thoroughly examined and discussed the association relationships, edges, angles, and reflection probabilities within the RIS-assisted UDN system. 
In this section, we present the ternary stochastic geometry theory, which is then used to analyze the system performance in the RIS-assisted UDN system.

\subsection{Ternary stochastic geometry theory}

Stochastic geometry can transform the mean form of random sums (Campbell's theorem) and random products (PGFL analysis) of point processes into integral form, so that it can simplify the computation and ensuring analytical tractability \cite{Stochastic_geometry_and_its_applications}.
However, the RIS-assisted UDN system exhibits a ternary network topology, the conventional Campbell's theorem and PGFL analysis cannot be directly applied. Therefore, we introduce the ternary stochastic geometry theory.
To begin with, we present Lemma \ref{lemma_1} and Lemma \ref{lemma_2} to derive the mean of the random sum/product of ternary stochastic geometry.

{\color{blue}
\begin{lemma}[Campbell's theorem of ternary stochastic geometry]\label{lemma_1}
	Two HPPPs $\Lambda_N$ and $\Lambda_M$ respectively with intensity $\lambda_n$ and $\lambda_m$ are two independent point processes on the two-dimensional plane $\mathbb{R}^2$. 
	Within this context, $f(x,y)$ is a measurable function on $\mathbb{R}^2$. 
	Thus, the mean value of the random sum can be derived as
		\begin{equation}
			{\small
			\begin{aligned}	
				  {\rm E} \left[  \sum_{ n  }^N \sum_m^M f(x_n,y_m)  \right] \overset{(a)}{=}& {\rm E} \left[ \sum_{n}^N 2 \pi \lambda_m\int f(x,y)y {\text d}y \right]\\
				  \overset{(b)}{=} & 4\pi^2\lambda_n\lambda_m\iint f(x,y) x y {\text d}x {\text d}y,
			\end{aligned}\label{Campbell}
			}
		\end{equation}%
		where $(a)$ and $(b)$ are because of the Campbell's theorem in the binary stochastic geometry.
\end{lemma}

Due to the nested function property of Campbell's theorem, it is only necessary to compute Campbell's Theorem twice to obtain the result stated in Lemma 1.
However, in the context of PGFL analysis, the PGFL analysis does not possess the nesting property found in the Campbell's theorem.
Consequently, it is particularly challenging to analyze the PGFL in ternary stochastic geometry. 
To address this difficulty and maintain the symmetry of the PGFL, we provide an approximate PGFL analysis in ternary stochastic geometry.

\begin{lemma}[approximate PGFL analysis of ternary stochastic geometry]\label{lemma_2}
	Two HPPPs $\Lambda_N$ and $\Lambda_M$ respectively with intensity $\lambda_n$ and $\lambda_m$ are two independent point processes on the two-dimensional plane $\mathbb{R}^2$. 
	Within this context, $f(x,y)$ is a measurable function on $\mathbb{R}^2$. 
	Thus, the mean value of the random product can be approximately derived as
    \begin{equation}
		{\small
		\begin{aligned}
			 & {\rm E} \left[ \prod_n^N \prod_m^M \!f(x_n,y_m) \!\right]  =  {\rm E} \left[  \left( f(x_n,y_m) \right)^{NM}  \right]\\
			 \overset{(a)}{=} & {\rm E} \left[   {\rm exp} \Big(  \!\!-2\pi\lambda_m\int(1 - (f(x_n,y_m))^N)y{\text d}y  \Big)  \right]\\
			 \overset{(b)}{\approx} & {\rm exp} \Big(  \!\!-2\pi\lambda_m \int(1 - {\rm E} \left[  (f(x_n,y_m))^N \right])y{\text d}y  \Big)\\
			\overset{(c)}{=} & {\rm exp} \Big(  \!\!-2\pi\lambda_m \!\int\!(1 \!- \! {\rm exp} (-\!2\pi \lambda_n\!\int\! (1\!-\!f(x,y))x{\text d}x)  )y{\text d}y  \Big)\\
			\overset{(d)}{\approx} & {\rm exp} \Bigg(\!\!-\!4\pi^2 \lambda_n \lambda_m \!\iint \!\!( 1\!-\! f(x,y))xy {\text d}x {\text d}y \Bigg),
		\end{aligned}}\label{PGFL}
	\end{equation}%
	where $(a)$ and $(c)$ are based on the PGFL analysis in binary stochastic geometry.
	$(b)$ is derived by exchanging the calculated order of the means.
	$(d)$ is obtained through the equivalent infinitesimal substitution of $1-{\rm exp}(-x) \approx x$.
\end{lemma}



In fact, Eq. $(c)$ of \eqref{PGFL} can also be viewed as an expression for the approximate PGFL analysis in ternary stochastic geometry. 
However, there are two reasons why it may not be an ideal choice: on the one hand, it lacks symmetry with respect to the two HPPPs; on the other hand, the nested function in Eq. $(c)$ is highly complex and does not facilitate subsequent analysis.}

\subsection{Statistical characteristics of signal and interference}

Based on Lemma \ref{lemma_1}, we derive the statistical characteristics of the signals and interference.
{\color{blue} In Appendix A, the detailed deduction processes for the power of the entire signal $D_{sum} = D_1 + D_2$ and the entire interference $I_{sum} = I_1 + I_2$ have been deduced.}
The resulting expressions for these quantities are respectively expressed as
	\begin{subequations}
		{\small
		\begin{align}
			{\rm{E}}\! \left[ \left|D_{sum} \right|^2  \right] 
			\!\!=\! & P_{tr}  \!\!\!\int_0^{\infty}  \!\! \left[ \! \frac{ 2\pi  \lambda_m  Q(  Q \!-\!\! 1 ) \Gamma^4(\varsigma\!+\!\! \frac{1}{2})}{\varsigma^2 \Gamma^4(\varsigma)}   \mathcal Q_{p}(\alpha, d^D_{1,1}) \right.  \notag \\
		& \!+\!\!  \frac{4 \pi^2 \lambda_m^2 Q^2\Gamma^4(\varsigma\!+\!\frac{1}{2})}{\varsigma^2 \Gamma^4(\varsigma)}   \mathcal Q_{p}^2(\frac{\alpha}{2}, d^D_{1,1})  \notag \\
		&\!+\!\! \frac{2\pi\lambda_m Q \Gamma^3(\varsigma\!+\! \frac{1}{2})}{\varsigma^{\frac{3}{2}} \Gamma^3(\varsigma)(1+ d^D_{1,1} )^{\frac{\alpha}{2}}} \mathcal Q_{p}(\frac{\alpha}{2}, d^D_{1,1}) \label{x_mean}  \\
		& \left. \! \!+\! \!\frac{ 1 }{(1\!+\!\! d^D_{1,1})^{\alpha}} \!\!+ \!\!2 \pi \lambda_m Q  \mathcal Q_{1}(\alpha, d^D_{1,1}) \right]\!\! f(d^D_{1,1})  {\text d} d^D_{1,1}\notag \\
		{\rm{E}} \!\left[ \left|I_{sum} \right|^2 \right] 
			\!= & 2 \pi      \lambda_{un}'P_{tr}\!\! \left[ \frac{1}{6} \!-\! \mathcal{F}\!\left( \alpha , \pi \lambda_{un}' \right) 
			\!+\! 2  \pi \lambda_m Q \!\!  \int_0^{\infty} \!\!\!\!   \int_{d_{1,1}^D}^{\infty}  \right. \notag\\
			& \!\!     \mathcal Q_{1}(\alpha, d^D_{1,1})  \left.  d_{n',1}^{D} {\text d} d_{n',1}^{D} f\left( d_{1,1}^D \right) {\text d}  d_{1,1}^D \right. \bigg],\label{I_mean}
		\end{align}}
	\end{subequations}%
where $\textstyle {\mathcal{F}}(a, k) \triangleq \int_0^{\infty} \frac{x}{(1+x)^{a}} e^{-k x^2} {\text d}x$, ${\mathcal{G}}(x) \triangleq \int_0^{\pi} \frac{{\rm P}(\beta_{n,m,1} = 1 | \Delta \theta^M_{n,m,1} ) }{\pi} x {\text d} \Delta \theta$, and 
\begin{equation}
		{\small
		\mathcal Q_{i}(\alpha, d^D_{1,1}) \! \triangleq \!\! \left\{
	    \begin{aligned}
		& \! \! \int^{\infty}_{0} \! \! \!\! \!{\mathcal{G}} \!\!\left(\!\! \frac{1 }{(\!1 \!\!+\!\! d_{m,1}^{R} \!)^{ \alpha }  (\!1 \!\!+\!\! d_{1,m}^{I} \!)^{\alpha }}  \!\!\right)\!\!  d_{1,m}^{I}  {\text d}  d_{1,m}^{I}	,i\!\!=\!\!1,\\
		& \!\!\int^{\infty}_{0} \!\!\!\!\!{\mathcal{G}} \!\!\left(\!\! \frac{p(\lambda'_{un}, d^I_{1,m})}{(\!1 \!\!+\!\! d_{m,1}^{R} \!)^{ \alpha }  (\!1 \!\!+\!\! d_{1,m}^{I} \!)^{\alpha }}  \!\!\right) \!\! d_{1,m}^{I}  {\text d}  d_{1,m}^{I}	,i\!\!=\!\!p,\\
		& \!\!\int^{\infty}_{0}\!\!\!\!\!\! {\mathcal{G}} \!\!\left(\!\! \frac{1-p(\lambda'_{un}, d^I_{1,m})}{(\!1 \!\!+\!\! d_{m,1}^{R} \!)^{ \alpha }  (\!1 \!\!+\!\! d_{1,m}^{I} \!)^{\alpha }}  \!\!\right)\!\!  d_{1,m}^{I}  {\text d}  d_{1,m}^{I}	,i\!\!=\!\!1\!\!-\!\!p.
		\end{aligned}\right. }
\end{equation}

Observing \eqref{x_mean},  it is found that $ {\rm{E}} [ |D_{\text{sum}} |^2 ] \ge {\rm{E}} [ |D_1 |^2 ] + {\rm{E}} [ |D_2 |^2 ]$, since the direct signal $D_1$ and the partially reflected signals of $D_{2}$ can be superimposed in the same phase.
Since $p(\lambda'_{un},d^I_{n,m})$ is a monotonically decreasing function of $\lambda'_{un}$, the number of RISs in the typical voronoi cell decreases as $\lambda'_{un}$ increases.
In this case, only a small number of RISs are used for beamforming, thereby diminishing the role of RIS within an extremely dense UDN system.
In \eqref{I_mean}, it becomes apparent that ${\rm{E}} [ |I_{sum} |^2 ] ={\rm{E}}[ I_{1} |^2 ] +{\rm{E}} [ |I_{2} |^2 ] $.
It is noted that $I_1$ and $I_2$ are independent of each other, and both follow complex Gaussian distributions with zero mean for their small-scale fading.
Furthermore, it is found that the power of $I_2$ is proportionality to $Q \lambda_m$. 
These results will be demonstrated in the subsequent simulation section.

\subsection{Performance analysis}

In this subsection, we derive approximate closed formulas of the coverage probability, area spectral efficiency (ASE), and energy efficiency of the RIS-assisted UDN system.
To begin with, we consider the power consumption model of the RIS-assisted system, which is based on existing literatures \cite{Power_Consumption,SG_fine_jinshi}. 
Typically, a linear approximation model is widely used to describe the power consumption of BS $P_n$ and RIS $P_m$, expressed as $P_n=\Delta_p P_{tr} + P_{n_s}$ and $P_m = Q P_{m_d} + P_{m_s}$. 
Here, $P_{n_s}$ and $P_{m_s}$ represent the static power consumption of the BS and RIS, respectively.
The dynamic power consumption of the BS and RIS are respectively defined by $\Delta_p P_{tr}$ and $Q P_{m_d}$, where ${1}/{\Delta_p}$ represents the power amplifier efficiency, and $P_{m_d}$ denotes the average regulation power consumption per reflection element of the RIS.

\subsubsection{Coverage probability}
The coverage probability $\text{P}_c(\delta)$ is defined as the probability of the typical UE's SINR $\gamma_1$, which is larger than or equal to the threshold $\delta$, expressed as
	\begin{equation}
		{\small
		\text{P}_c(\delta)=1-\text{P}_o(\gamma_1 < \delta)=\text{P}(\gamma_1 \ge \delta),}
	\end{equation}%
where $\text{P}_o( \gamma_1 < \delta)$ represents the outage probability.

The power of the entire signal $D_{sum}$ can be given as
	\begin{equation}
		{\small
		\label{D_eq}
		\begin{aligned}
			&|D_{sum}|^2=|D_1+D_{2}|^2 \\
			\approx & \frac{|h_{1,1}|^2P_{tr}}{(1 + d^D_{1,1})^{\alpha}} +  \sum_{m=1}^M \frac{\beta^2_{1,m,1}|h^C_{1,m,1}|^2P_{tr}}{(1 + d^R_{m,1})^{\alpha} (1 + d^I_{1,m})^{\alpha} } + \epsilon_1 + \epsilon_2, 
		\end{aligned}}
	\end{equation}%
where $\epsilon_1$ and $\epsilon_2$ are modified factors to simplify the analysis processing, and are respectively expressed as
	\begin{subequations}
		{\small
		\begin{align}
			\epsilon_1 \! = \! & \! \! \sum_{m=1}^M \frac{ 2 \Gamma(\varsigma +1) \beta_{1,m,1}  \sigma_h |h^C_{1,m,1}|P_{tr}}{ \sqrt{\varsigma}\Gamma(\varsigma)(1 + d^D_{1,1})^{\frac{\alpha}{2}} (1 + d^R_{m,1})^{\frac{\alpha}{2}} (1 + d^I_{1,m})^{\frac{\alpha}{2}} } ,\\
			\epsilon_2 \! =\! &\! \!  \sum_{m=1}^M  \!\sum_{\substack{ m' \neq m  \\ m'=1}}^M  \!\!\!\frac{ 2 \beta_{1,m,1} \beta_{1,m',1} |h^C_{1,m,1} | | h^C_{1,m',1}| P_{tr}}{(\!1 \!\!+\!\! d^R_{m,1})^{\frac{\alpha}{2}} (\!1\! \!+\!\! d^I_{1,m})^{\frac{\alpha}{2}}  (\!1\!\! + \!\!d^R_{m',1})^{\frac{\alpha}{2}} (\!1 \!\!+\!\! d^I_{1,m'})^{\frac{\alpha}{2}}  }   .
		\end{align}}
	\end{subequations}

In the RIS-assisted UDN system, the approximate closed formula for coverage probability can be deduced as 
\begin{equation}
	{\small
	\label{coverage_hebing}
	\begin{aligned}
		 \text{P} \!\left( \gamma_1 \!\ge\! \delta \right)\!
		\overset{(a)}{\approx} & {\rm E}\Big[   \frac{\Gamma[ \varsigma + 1,\varsigma  \mathcal X ]}{\Gamma[\varsigma + 1]}   \Big]
		\!\overset{(b)}{\le} \! 1\! - \!{\rm E}\Big[  \Big( \!1- \!{\rm exp}\Big(  \!\!-\! \varpi \mathcal X   \Big) \Big)^\varsigma \Big]\\
		= & \sum_{k=1}^\varsigma C_\varsigma^k (-1)^{k+1} {\rm E}\left[  {\rm exp}\Big(  - k \varpi  \mathcal X   \Big)   \right],
\end{aligned}}
\end{equation}%
where $\mathcal X \!= \! \delta \Big(  \frac{ ( 1 \!+\! d^D_{1,1})^{\alpha} \sigma_n^2}{P_{tr}}+\sum\limits_{n'=2}^N \!\sum\limits_{m=1}^M \!\frac{  ( 1+ d^D_{1,1})^{\alpha}  \beta^2_{n',m,1}|h^C_{n',m,1}|^2 }{( 1 \!+\! d^I_{n',m})^{\alpha}( 1 \!+\! d^R_{m,1} )^{\alpha}} + \sum\limits_{n'=2}^N \!\!\frac{  ( 1\!+\! d^D_{1,1})^{\alpha}  |h_{n',1}|^2 }{(  1\!+\!d^D_{n',1} )^{\alpha}} 
   \Big) \!-\!  \sum\limits_{m=1}^M \!\!\frac{ ( 1\!+\! d^D_{1,1})^{\alpha}  \beta^2_{1,m,1} |h^C_{1,m,1}|^2   }{(1 \!+\! d^R_{m,1})^{\alpha} (1 + d^I_{1,m})^{\alpha}}
   - \frac{ ( 1\!+\! d^D_{1,1})^{\alpha} \epsilon_1}{P_{tr}} - \frac{ ( 1\!+\! d^D_{1,1})^{\alpha} \epsilon_2}{P_{tr}}  $  and $\varpi=\frac{ {\rm min} \{1, ((1+ \varsigma)!)^{-\frac{1}{\varsigma+1}} \} \varsigma}{1+ \varsigma }$.
   $(a)$ is derived by computing the probability using the Nakagami distribution for $h_{1,1} \sim Nakagami (\varsigma,1)$.
   $(b)$ is established using inequalities involving the incomplete gamma function \cite{GAMMA}.
   Furthermore, ${\rm E}\left[  {\rm exp}\Big(  - k \varpi \mathcal X   \Big)   \right]$ is expressed as
\begin{equation}
	{\small
	\begin{aligned}
		 & {\rm E}\left[  {\rm exp}\Big(  - k \varpi \mathcal X   \Big)   \right]\\
		\!\overset{(a)}{\approx}\!\! &
		\int_0^{\infty}{\rm exp}\Bigg(-  \frac{k \varpi(1 + d^D_{1,1})^4\delta\sigma_n^2}{P_{tr}} \Bigg)\\
		&\underbrace{{\rm exp}\Bigg( -2\pi\lambda_m \int_0^{\infty} (1 - \mathcal L_{\epsilon_1} \mathcal L_{D_2}  ) d^R_{m,1} {\text d} d^R_{m,1} \Bigg)}_{f_{D_2}(\lambda_m,Q) }\\
		&  \underbrace{{\rm exp}\Bigg( \!\!\!-\!4\pi^2\lambda_m\lambda'_{un} \!\!\int_0^{\infty}\!\!\!\!\! \int_0^{\infty} \!\!\!\!\!\!(1\!\!-\!\!\mathcal L_{I_2} \mathcal L_{\epsilon_2}) d^R_{m,1} d^D_{n',1} {\text d} d^R_{m,1} {\text d} d^{D}_{n',1} \!\Bigg)}_{f_{I_2,\epsilon_2} (\lambda_m, Q, \lambda'_{un})}\\
		&\underbrace{{\rm exp} \Bigg( -2\pi \lambda'_{un}\int_0^{\infty}(1 - \mathcal L_{I_1}) d^D_{n',1} {\text d} d^D_{n',1}  \Bigg)}_{f_{I_1}(\lambda'_{un})} f(d^D_{1,1}) {\text d}d^D_{1,1}
		,
\end{aligned}}
\end{equation}%
where $(a)$ is because of Lemma \ref{lemma_2}, and ${\mathcal L}_x $ denotes the Laplace transform of the signal $x$. The specific expressions and the proof are given in Appendix B. 

{\color{blue} It is been observed that $f_{D_2}(\lambda_m, Q)$ may satisfy $f_{D_2}(\lambda_m, Q) \le 1$, while both $f_{I_1}(\lambda'_{un})$ and $f_{I_2,\epsilon_2}(\lambda_m, Q, \lambda'_{un})$ are less than $1$.
Hence, when $\lambda'_{un}$ is small, appropriately increasing $\lambda_m$ and $Q$ can enhance the coverage. 
However, when $\lambda'_{un}$ is large, the impact of such adjustments may not yield the same results and even lead to a decrease in coverage.
This finding aligns with the conclusions obtained from the statistical signal analysis, where an increase in $\lambda'_{un}$ weakens the signal enhancement capability of the RIS. 
Additionally, the interfering signal power increases as $\lambda'_{un}$ grows.
Therefore, it is important to carefully design the relationship between $\lambda_m$, $Q$, and $\lambda'_{un}$ to optimize system performance.}

\subsubsection{Area Spectral Efficiency}
In general, ASE is widely used as an important performance metric of an UDN.
The ASE is defined as the sum rate per bandwidth per area, expressed as
\begin{equation}
	{\small
	\begin{aligned}
	\label{ASE_UDN}
	\eta_s = & \frac{1}{\pi R^2} \sum_{n=1}^{N} {\rm E} \left[  \rm{log}_2 \left(1+ \gamma_n \right) \right]
	=  \lambda'_{un} {\rm E} \left[\rm{log}_2 \left(1+ \gamma \right) \right].
	\end{aligned}}
\end{equation}

In the RIS-assisted UDN, the approximate closed formula of ASE can be deduced as 
	\begin{equation}
		{\small
		\label{ase_frist}
		\begin{aligned}
			\eta_s 
			\!\approx & \frac{\lambda'_{un}}{{\text ln}2}\int_0^{\infty} \!\!\! \int_0^{\infty}
			\!\!{\rm exp} \Bigg( -2\pi \lambda'_{un}\int(1 - \mathcal L_{I_1}(z) ) d^D_{n',1} {\text d} d^D_{n',1}  \Bigg)\\
			&\Bigg[ {\rm exp}\Bigg( \!\!-\! 4\pi^2\!\lambda_m\lambda'_{un} \!\!\! \iint \!\!(1\!\!-\!\!\mathcal L_{I_2}(z)) d^R_{m,1} d^D_{n',1} {\text d} d^R_{m,1} {\text d} d^{D}_{n',1} \!\Bigg) \\
			& -\!\mathcal L_{D_1}(z) 
			{\rm exp}\Bigg( \!\!-\!2\pi\lambda_m \int (1 \!- \!\mathcal L_{\epsilon_1}(z) \mathcal L_{D_2}(z)  ) d^R_{m,1} {\text d} d^R_{m,1} \Bigg) \\ 
			& {\rm exp}\Bigg( \!\!\!-\! 4\pi^2\!\lambda_m\lambda'_{un} \!\!\! \iint \!\!(1\!\!-\!\!\mathcal L_{I_2}(z) {\mathcal L}_{\epsilon_2}(z) ) d^R_{m,1} d^D_{n',1} {\text d} d^R_{m,1} {\text d} d^{D}_{n',1} \!\!\Bigg) \! \Bigg]		\\
			&	\frac{{\rm exp }(- z \sigma^2_n )}{z}  f(d^D_{1,1}) {\text d}d^D_{1,1} {\text d}z
		\end{aligned}}
	\end{equation}
where the specific expressions and the proof are given in Appendix C.


\subsubsection{Area Energy Efficiency}

The AEE is generally used to measure the energy consumption of a network, which is defined as the ratio of the ASE to the corresponding power consumption \cite{interference_RIS_intelligent,SG_fine_jinshi,AEE}, expressed as
	\begin{equation}
		{\small
		\label{eta_e}
		\eta_e = \frac{\eta_s}{ \lambda'_{un} (\Delta_p P_{tr} + P_{n_s} ) + \lambda_m (Q P_{m_d} +  P_{m_s})  },}
	\end{equation}%
where the unit scale of $\eta_e$ is the bit /Joule/unit area.



\subsubsection{Energy Coverage Efficiency}

To provide a more comprehensive analysis and description of the coverage reconfigurability in RIS-assisted systems, we introduce the concept of energy coverage efficiency (ECE). 
ECE is defined as the ratio of the coverage probability to the corresponding power consumption, with the unit scale of coverage area per Watts. 
Mathematically, it can be expressed as
\begin{equation}
	{\small
	\label{eta_c}
	\begin{aligned}
		\eta_c &= \frac{\pi R^2 \text{P} \left( \gamma_{1} \ge \delta \right)  }{N \bar{P}_n }= \frac{\text{P} \left( \gamma_{1} \ge \delta \right)  }{ \lambda'_{un} \bar{P}_n }\\
		 &=\frac{  \text{P} \left( \gamma_{1} \ge \delta \right)  }{ \lambda'_{un} ( \Delta_p P_{tr} + P_{n_s} ) + \lambda_m (Q P_{m_d} + P_{m_s}) },
	\end{aligned}}
\end{equation}%
where $ \bar{P}_n =  P_n + \frac{\lambda_m}{\lambda'_{un}} P_m$ is the normalization power consumption per BS.
The conception of ECE is used to descript the effect of RISs to a network.

\section{Simulation Results}
In this section, we compare the performance of the RIS-assisted UDN system with the classical UDN system.
{\color{blue}To ensure accuracy, we refer to relevant and authoritative hardware literature \cite{Power_Consumption,how_much} for the power consumption of the BS and the power consumption parameters of the RIS.
In this paper, we assume the ideal working condition of the RIS, ignoring the effect of reflection efficiency.
For a fair comparison, we set the power amplifier efficiency of the BS to an ideal value of $\Delta_p = 1$.
The noise power is set to be $\sigma_n^2 = N_0 B = 7.96 \times 10^{-14}$ W.
The path loss factor $\alpha$ is $4$. 
For small-scale fading, we consider two scenarios: $ \varsigma = 1$ (Rayleigh distribution) and $ \varsigma = 10$ (approximate Rice distribution) \cite{3Gpp}.
To compare the effect of RIS deployment with different parameters under the same power consumption, we consider two parameter control groups. One group is $\lambda_m=0.1, Q=10$ and $\lambda_m=0.05, Q=563$, while the other group is $\lambda_m=0.01, Q=10$ and $\lambda_m=0.005, Q=563$. }
For further details on the simulation conditions, please refer to Table \ref{simulation_parameters}.

\begin{table}[t]
    \begin{center}
	      \caption{Simulation conditions.}
    \begin{tabular}{p{0.25\textwidth}<{\centering} p{0.19\textwidth}<{\centering}}
	\toprule
	System parameters & Values  \\
	\midrule
	Transmit power $P_{tr}$ & $30$ dBm \cite{how_much} \\
	Static power consumption of BS $P_{n_s}$ & $14.7$ W \cite{how_much} \\
	Power amplifier efficiency $ {1}/{\Delta_p}$ & 1 \\
	Noise power spectral density $N_0$ & $-174$ dBm/Hz \\
	Bandwidth $B$ & $20$ MHz \\
	Path loss factor $\alpha$  & $4$ \\
	Small scale fading  $Nakagami(\varsigma,1)$  & $1$, $10$ \\
	Radius of the given area $R$ & $1000$ m \\
	Intensity of active BS $\lambda_{un}'$ & $ 0.001 \sim 10$ (${\rm m}^{\!-\!2}$) \\
	Intensity of RIS $\lambda_m$ &  $0.005$, $0.01$, $0.05$, $0.1$ (${\rm m}^{\!-\!2}$)\\
	Number of reflection elements $Q$ & $10$, $563$ \\
	Dynamic power consumption of reflection element $P_{m_d}$ & $12$ mW \cite{Power_Consumption}\\
	Static power consumption of RIS $P_{m_s}$ & $6.52$ W \cite{Power_Consumption} \\
	Communication threshold $\delta$ & $-30 \sim 50$ dB \\
	\bottomrule
	\label{simulation_parameters}
 \end{tabular}
\end{center}
\end{table}


Fig. \ref{tongji} illustrates the average power of signals and interference, with RIS intensities $\lambda_m =0.1$, $0.01$, $0.05$ {\color{blue}and} $0.005$, accompanied by corresponding RIS reflection element numbers of $Q=10$, $10$, $563$ and $563$.
As the active BS intensity $\lambda'_{un}$ increases, the average power of signals initially rises and then decreases. 
The reason is that, when the RIS intensity $\lambda_m$ remains constant, a higher active BS intensity leads to smaller voronoi cells, reducing beamforming efficiency of the RIS.
Taking $\lambda_m=0.005$ and $Q=563$ as an example, the average power of $D_{sum}$ reaches its maximum value at $\lambda'_{un}=1$. 
For $\lambda'_{un}\le 1$, we attribute the limitation to the active BS intensity $\lambda'_{un}$, while for $\lambda'_{un}\ge 1$, it is limited by the RIS intensity $\lambda_{m}$.
In the Nakagami channel, a larger  value of $\varsigma$ results in increased signal power. Comparing $\varsigma=1$ and $\varsigma=10$, there is a maximum channel gain difference of approximately $1.9$ dB.
With the same RIS power consumption, a larger number of reflection elements (represented by $Q$) can enhance the average power of $D_{sum}$ due to the higher beamforming gain. For instance, when $\lambda'_{un}=0.1$ and $\varsigma=10$, the case of $\lambda_m=0.005$ with $Q=563$ exhibits an approximately $24.7$ dB improvement compared to the case of $\lambda_m=0.01$ with $Q=10$. 
\begin{figure}[t]
    \centering
	\includegraphics[width=2.6in]{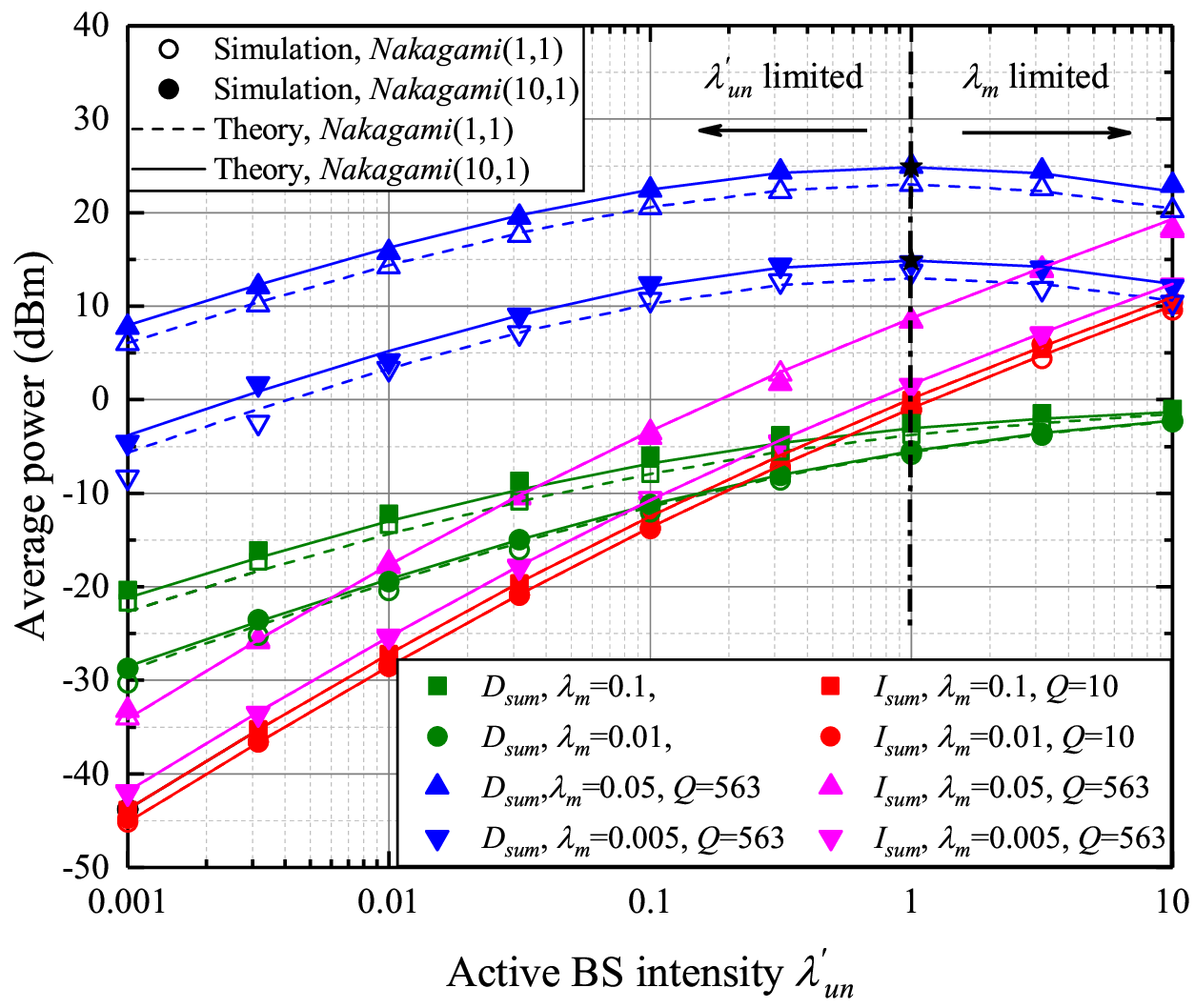}	
	\caption{The average power of signals and interference in the case of $\lambda_m=0.1$, $0.01$, $0.05$ and $0.005$.}
	\label{tongji}
\end{figure}


Fig. \ref{cdf} displays the outage probabilities of the typical UE, where the active BS intensity is set to be $\lambda'_{un}=0.01$.
In general, the UDN system with RIS exhibits improved outage probabilities for the typical UE compared to the classical UDN system. 
However, under a low communication threshold $\delta$, the UDN system with RIS performs worse than that of the classical UDN system. 
This can be attributed to the fact that the RIS simultaneously enhances both the signals and interference for cell-edge UEs with low SINR. 
We refer to this phenomenon as the Matthew's effect of the RIS-assisted UDN system.
With the same RIS power consumption, it is evident that the case of $\lambda_m=0.05$ and $Q=563$ provides significant assistance to cell-center UEs with high SINR, while the case of $\lambda_m=0.01$ and $Q=10$ offers better service for cell-edge UEs with low SINR.
For instance, when $\varsigma=10$ and $\gamma_1 = 30$ dB, the outage probability of the case of $\lambda_m=0.05$ and $Q=563$ is approximately $72.6 \%$ of the case with $\lambda_m=0.01$ and $Q=10$. 
On the other hand, when $\varsigma=10$ and $\gamma_1 = -10$ dB, the outage probability of the case of $\lambda_m=0.05$ and $Q=563$ is around $20.8$ times larger than the case of $\lambda_m=0.01$ and $Q=10$. 
The reason is that a larger number of reflection elements can significantly enhance the signals of the cell-center UEs with high SINR.
\begin{figure}[t]
    \centering
	\includegraphics[width=2.6in]{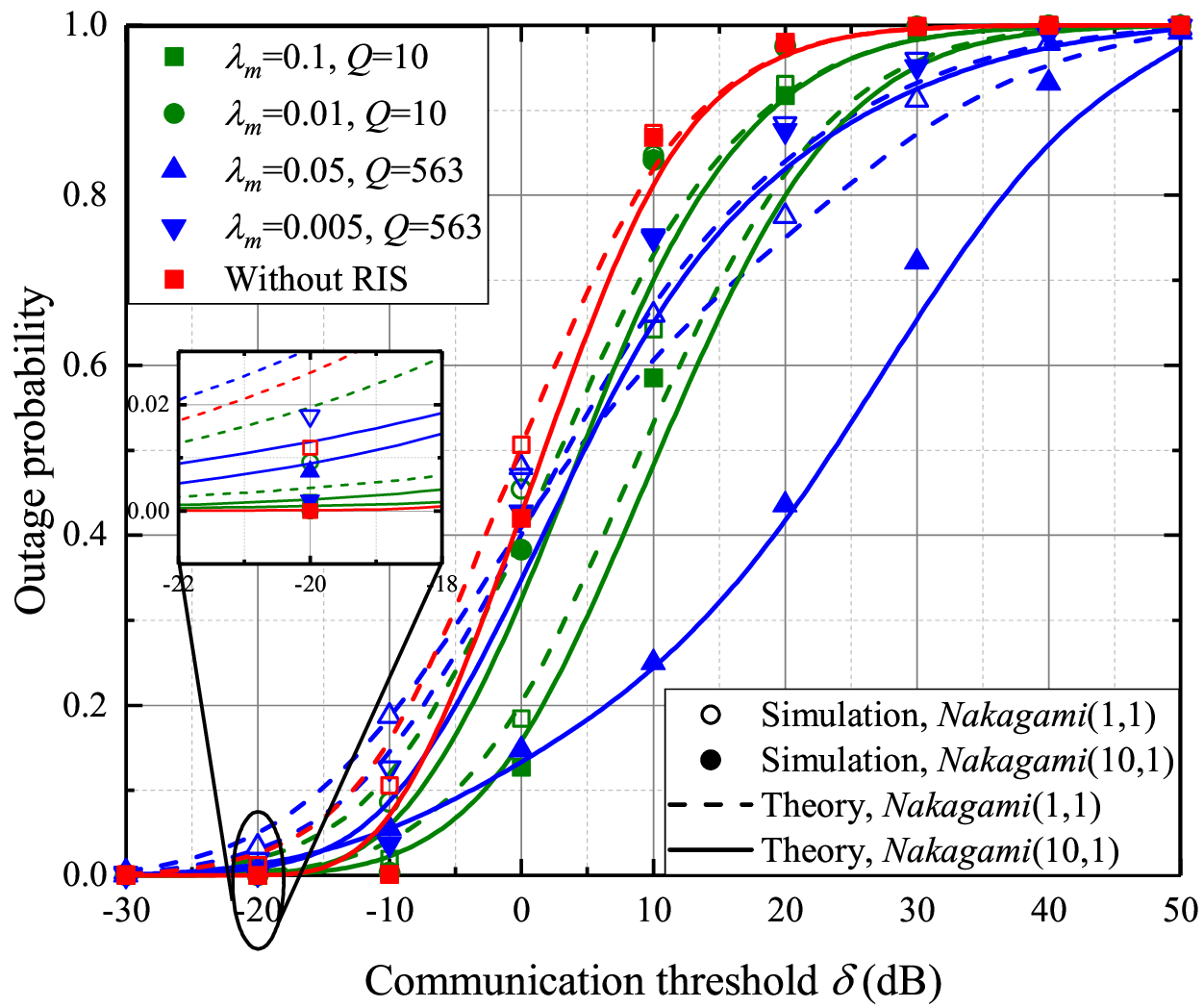}	
	\caption{The outage probability of the typical UE with $\lambda'_{un}=0.01$.}
	\label{cdf}
\end{figure}


The coverage probabilities of the RIS-assisted UDN and classical UDN are depicted in Fig. \ref{cover}, with RIS intensities $\lambda_m =0.1$, $0.01$, $0.05$, and $0.005$, and corresponding RIS reflection element numbers $Q=10$, $10$, $563$, and $563$.
As the active BS intensity $\lambda'_{un}$ increases, the coverage probability exhibits a decreasing trend due to the increased interference among BSs.
Overall, the RIS-assisted UDN system provides significantly better coverage probability compared to the case without RIS. 
For instance, when $\lambda'_{un}= 0.001$ and $\varsigma =1$, the coverage probability of the case of $\lambda_m=0.05$ and $Q=563$ is improved {\color{blue} by} about $45.1\%$  than the case without RIS.
When there are more reflection elements, the coverage probability experiences more rapid variations in response to changes in the active BS intensity $\lambda'_{un}$, as it may introduce more interference. 
For example, when $\lambda'_{un} = 0.001$ and $\varsigma=1$, the case of $\lambda_m=0.005$ and $Q=563$ performs about $18.1\%$ better than the case of $\lambda_m=0.01$ and $Q=10$. 
Conversely, when $\lambda'_{un} = 0.1$ and $\varsigma=1$, the case of $\lambda_m=0.005$ and $Q=563$ shows approximately $4.7\%$ worse coverage probability than the case of $\lambda_m=0.01$ and $Q=10$.
Moreover, in the Nakagami channel, a larger value of $\varsigma$ can improve channel gain and subsequently enhance coverage, particularly when $\lambda'_{un}$ is small. 
However, as $\lambda'_{un}$ increases and interference becomes more substantial, this effect diminishes and can even become counterproductive.
Additionally, the simulation results align well with the theoretical curves.
\begin{figure}[t]
    \centering
	\includegraphics[width=2.6in]{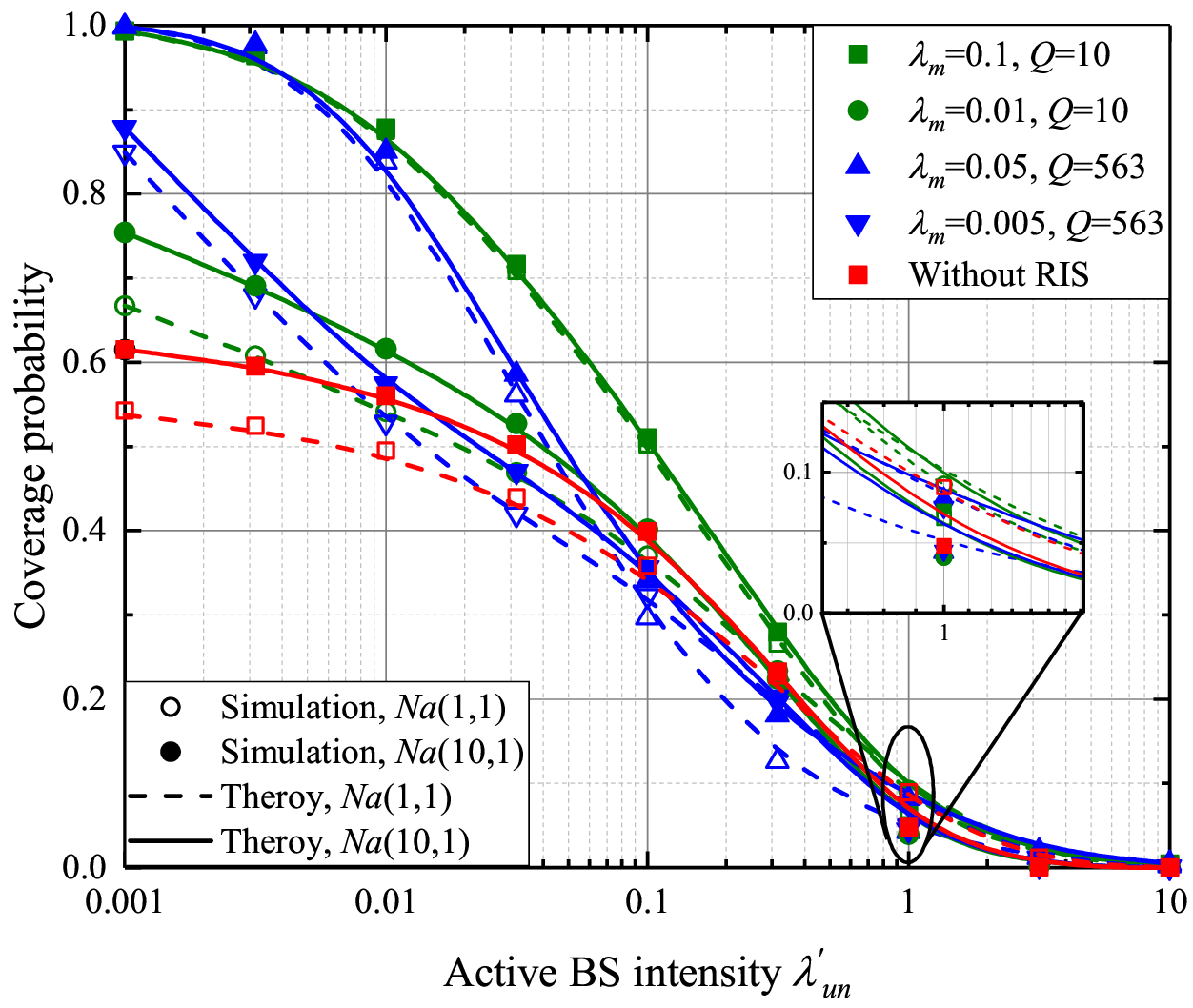}	
	\caption{The coverage probability of RIS-assisted UDN system, where $\lambda_m=0.1$, $0.01$, $0.05$ and $ 0.005$.}
	\label{cover}
\end{figure}


Fig. \ref{ase} illustrates the ASE performances of the RIS-assisted UDN and classical UDN systems, with RIS intensities $\lambda_m=0.1$, $0.01$, $0.05$, and $0.005$.
Compared to the classical UDN system, the RIS-assisted UDN system demonstrates higher ASE. 
For instance, when $\lambda'_{un}=0.001$, $\varsigma=1$, $\lambda_m=0.05$, and $Q=563$, the RIS-assisted UDN system achieves approximately $442.1\%$ higher ASE than the classical UDN system.
With the same RIS power consumption, the ASE performance follows a similar trend to the coverage probability. 
However, there is a performance inversion as the active BS intensity $\lambda'_{un}$ increases. 
For example, when $\lambda'_{un}=0.001$ and $\varsigma=1$, the case of $\lambda_m=0.05$ and $Q=563$ performs approximately $96.1\%$ better than the case of $\lambda_m=0.1$ and $Q=10$. 
However, when $\lambda'_{un}=10$ and $\varsigma=1$,, the case of $\lambda_m=0.05$ and $Q=563$ exhibits around $12.3\%$ worse ASE than the case of $\lambda_m=0.1$ and $Q=10$.
The simulated results align well with the theoretical results.
\begin{figure}[t]
    \centering
	\includegraphics[width=2.6in]{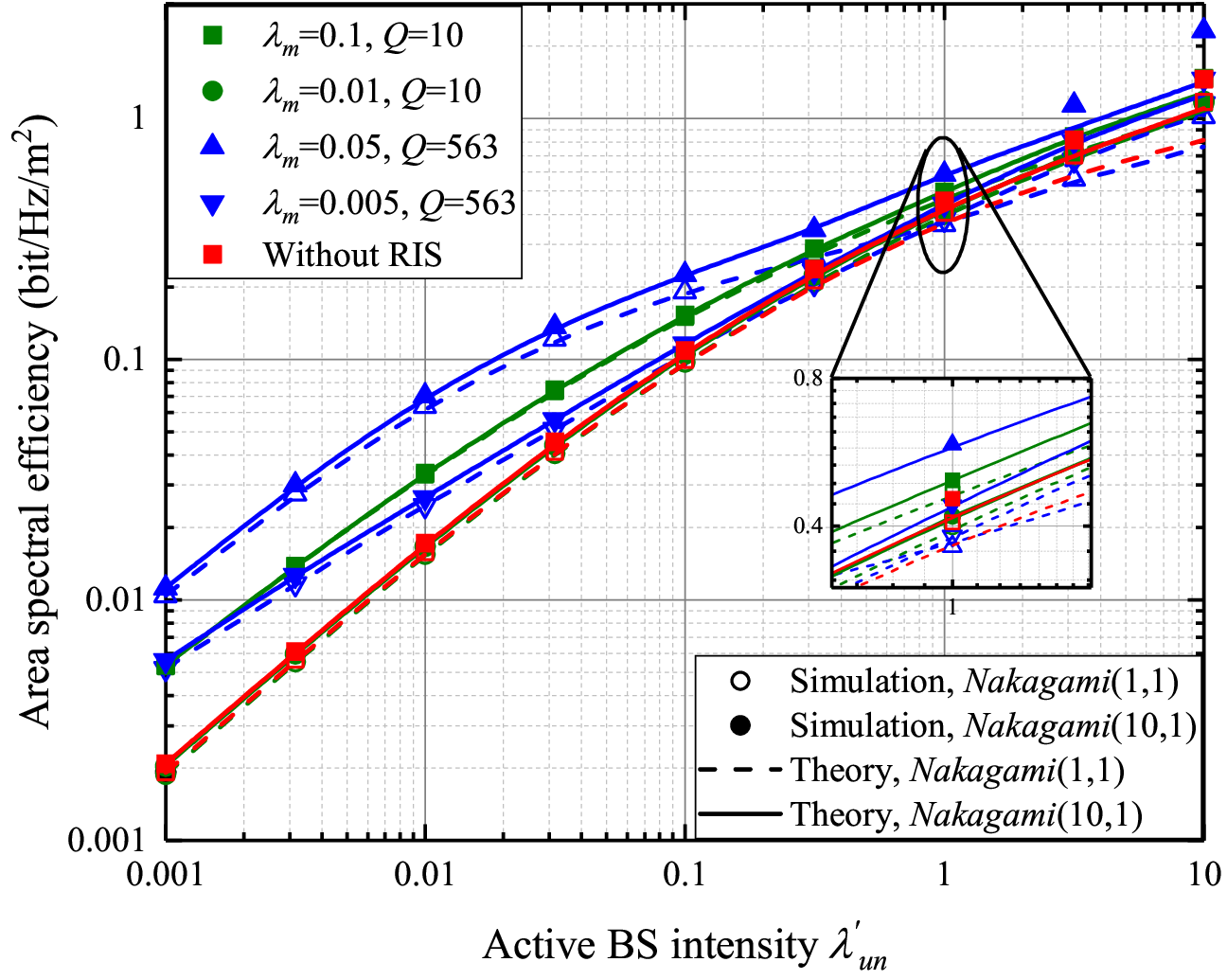}	
	\caption{The ASE of RIS-assisted UDN system with $\lambda_m=0.1$, $0.01$, $0.05$  and $0.005$.}
	\label{ase}
\end{figure}


The AEE performances of the RIS-assisted UDN system and the classical UDN system are illustrated in Fig. \ref{AEE1}, with the RIS intensities $\lambda_m=0.1$, $0.01$, $0.05$, and $0.005$.
For a given active BS intensity $\lambda'_{un}$, it is always possible to find an optimal RIS scheme that outperforms the classical UDN system. 
As the active BS intensity $\lambda'_{un}$ increases, the AEE performance of the RIS-assisted UDN system exhibits an extreme value due to the significant energy consumption from the RISs, which increases linearly. 
When $\lambda'_{un}$ is larger, the energy consumption mainly comes from the BSs; on contrast, when $\lambda'_{un}$ is smaller, the energy consumption comes mainly from the RISs.
Taking $\lambda_m=0.05$ and $Q=563$ as an example, when $\lambda'_{un}=0.1$ and $\varsigma=1$, the AEE performance of the RIS-assisted UDN system is approximately $32.9\%$ better than the classical UDN system. 
However, when $\lambda'_{un}=0.01$ and $\varsigma=10$, the RIS-assisted UDN system performs around $28.5\%$ worse than the classical UDN system.
Therefore, the configuration of the RIS should be optimized to match the active BS intensity, achieving a balance between energy consumption and spectral efficiency.
\begin{figure}[t]
    \centering
	\includegraphics[width=2.6in]{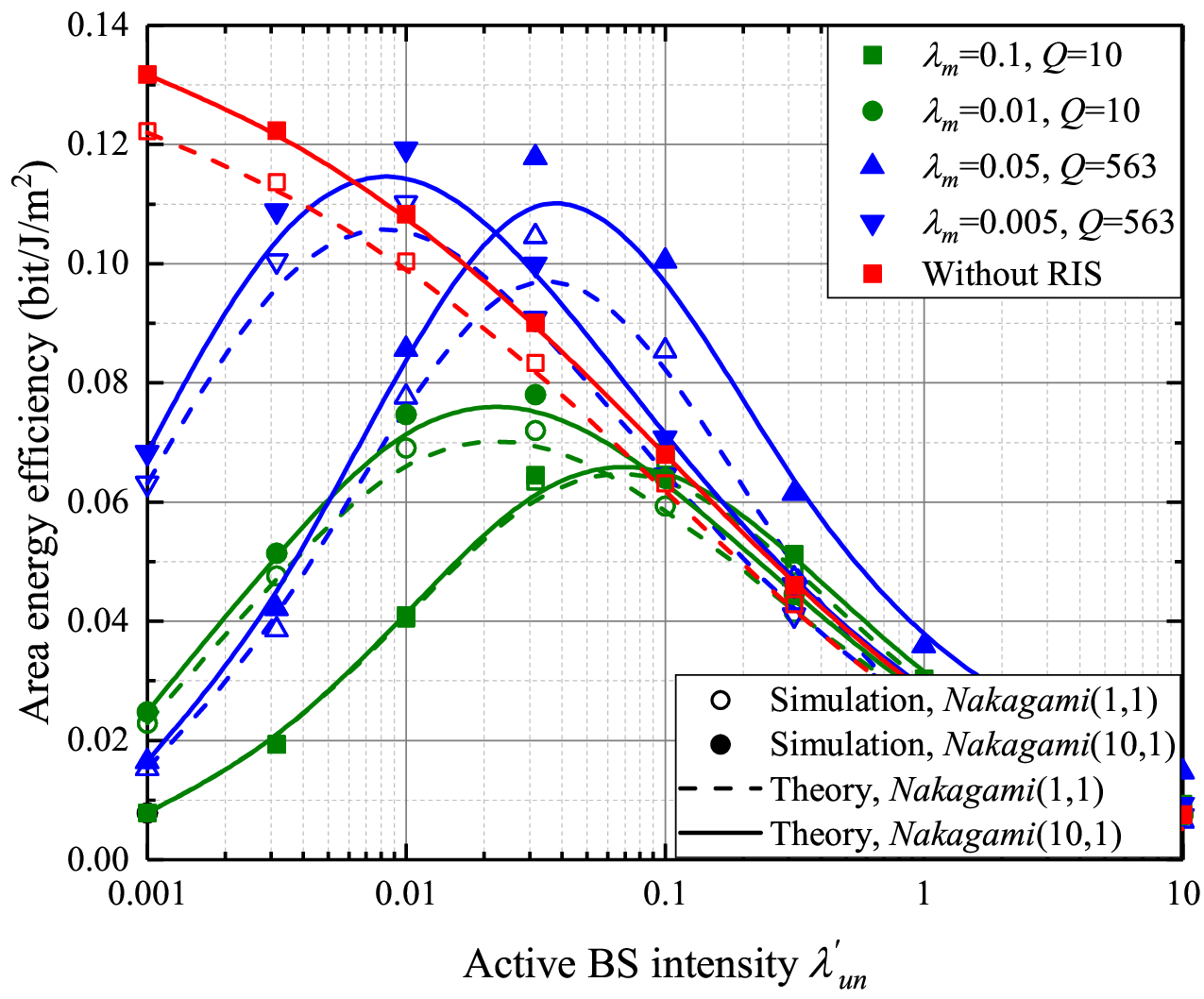}	
	\caption{The AEE of RIS-assisted UDN system with $\lambda_m=0.1$, $0.01$, $0.05$ and $ 0.005$.}
	\label{AEE1}
\end{figure}


The ECE performance of the RIS-assisted UDN system and classical UDN systems is illustrated in Fig. \ref{ECE}, with the RIS intensities $\lambda_m=0.1$, $0.01$, $0.05$, and $0.005$.
In general, the ECE decreases as the intensity of active BSs $\lambda'_{un}$ increases, due to the reduced coverage area of a dense network.
Furthermore, the ECE of the RIS-assisted UDN system exhibits greater robustness compared to the classical UDN system. 
For instance, when $\varsigma=1$, transitioning from $\lambda'_{un}=0.01$ to $\lambda'_{un} = 3.16$, the classical UDN system experiences a change of approximately $44.9$ dB, whereas the RIS system with $\lambda_m=0.1$ and $Q=10$ only undergoes a change of about $35.9$ dB. 
This is attributed to the RIS's ability to enhance the signal effect, thereby mitigating the loss in coverage area.
\begin{figure}[t]
    \centering
	\includegraphics[width=2.6in]{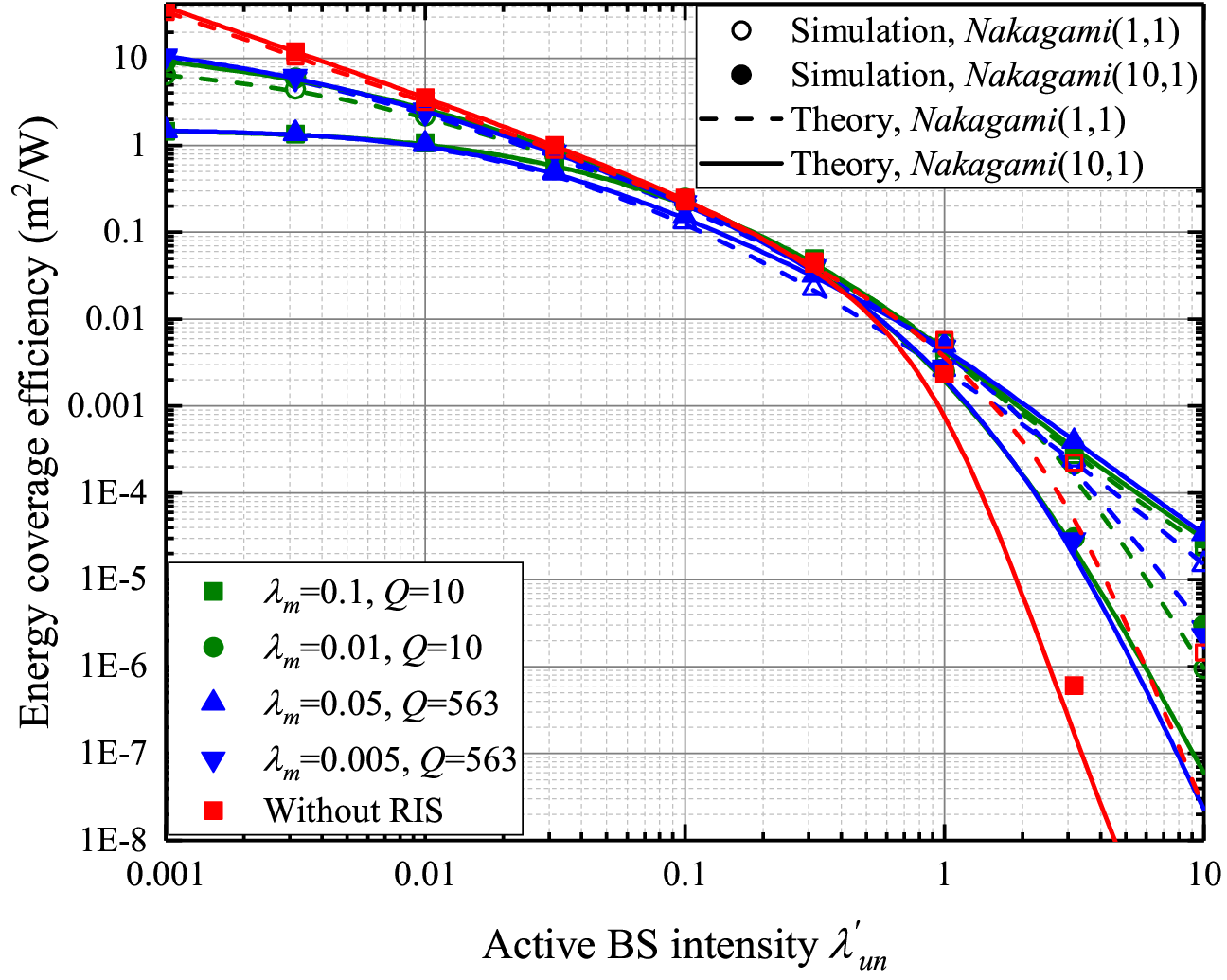}	
	\caption{The ECE of the RIS-assisted UDN system with $\lambda_m=0.1$, $0.01$, $0.05$ and $0.005$.}
	\label{ECE}
\end{figure}

{\color{blue}
\section{Conclusion}
In this paper, we present the ternary stochastic geometry theory to analyze the coupling and collaboration among multiple nodes in complex network systems, e.g.,  RIS-assisted UDN systems.
We first propose a dual coordinate analysis method which facilitates the analysis of interrelationships among multiple point processes.
Based on the geometric properties of the typical $\triangle {\rm BRU}$, we examine correlations and derive PDFs for correlation edges and correlation angles. Additionally, we investigate the distribution of RIS reflection states.
Furthermore, we introduce the Campbell's theorem and approximate PGFL analysis to compute the random sum and random product of the ternary stochastic geometry.
We conduct a comprehensive analysis of the system performances of RIS-assisted UDN, including signal statistical characteristics, coverage probability, ASE, AEE, and ECE.
Simulation results show that the RIS-assisted UDN system can be theoretically designed to outperform the classical UDN system in all performance aspects.
Random deployment of RIS exhibits the Matthew effect, improving the SINR for the cell center UEs while potentially weakening the SINR for the cell edge UEs.
Moreover, for the same RIS power consumption, a deployment scheme with higher RIS strength is more suitable for dense UDN systems, while a scheme with a higher number of RIS reflection elements is more suitable for sparse UDN systems. 
This finding motivates further research on the specific relationship between communication requirements and RIS configuration.
}
\begin{appendices}

	\section{}

	Recalling Section II and III, 
	the reflected signals $D_2$ is further given as
	\begin{small}
		\begin{equation}
			\begin{aligned}
				D_2
				\!= \! &\!\!\!\! \underbrace{\sum_{m \in \Psi^M_{1}}  \!\!\!  \frac{h_{1,m,1}   \beta_{1,m, 1} s_1 }{ {  (1 \!\!+\! \!d_{m,1}^{R} )^{{\frac{\alpha}{2}}} (1\!\!+\!\! d_{1,m}^{I} )^{{\frac{\alpha}{2}}} }} }_{ D_{2,\Psi^M_{1}} }
				\!+\! \!\!\!  \!\underbrace{ \sum_{m \in \widetilde{  \Psi}^M_{1}}\!\!\!   \frac{h_{1,m,1}   \beta_{1,m, 1} s_1 }{ {  (1 \!\!+\!\! d_{m,1}^{R} )^{{\frac{\alpha}{2}}}  (1\!\!+ \!\!d_{1,m}^{I} )^{{\frac{\alpha}{2}}} }} }_{ D_{2,\widetilde{  \Psi}^M_{1}}} ,
			\end{aligned}
		\end{equation}
	\end{small}%
	where $D_{2,\Psi^M_{1}} $ and $ D_{2,\widetilde{  \Psi}^M_{1}} $ respectively represent the reflected signals from the RISs belonging to the typical BS, and the rest RISs.
	
	\setcounter{equation}{28}
	\begin{figure*}[t]
		\begin{footnotesize}
			\begin{equation}
				\begin{aligned}
					 {\rm E}\left[  {\rm exp}\Big(  - k \varpi \mathcal X   \Big)   \right]
					= & 
					{\rm E} \! \Bigg[ {\rm {exp}} \left(\!\! -\! \frac{ k \varpi( 1\!+\! d^D_{1,1})^{\alpha} \delta \sigma_n^2}{ P_{tr}}   \!\right)  \!\!\prod_{n'=2}^N  \!\! \underbrace{{\mathcal L}_{I_1}\!\! \left( \!\frac{ k \varpi(1 \!+ \!d^D_{1,1})^{\alpha} \delta }{ (  1\!+\!d^D_{n',1} )^{\alpha}} \!\right)}_{\mathcal L_{I_1}}  
					\!\prod_{m=1}^M \! \underbrace{{\mathcal L}_{\epsilon_1}\! \left(  \frac{ - k \varpi 2 \Gamma( \varsigma+1) \beta_{1,m,1} (1 + d^D_{1,1})^{\frac{\alpha}{2}}  }{ \sqrt{\varsigma}\Gamma(\varsigma) (1 + d^R_{m,1})^{\frac{\alpha}{2}} (1 + d^I_{1,m})^{\frac{\alpha}{2}} }  \right)}_{\mathcal L_{\epsilon_1}}  \\ 
					&\prod_{m=1}^{M} \underbrace{{\mathcal L}_{D_2} \left( \frac{ -k \varpi \beta^2_{1,m,1} (1 + d^D_{1,1})^{\alpha}  }{  (1 + d^R_{m,1})^{\alpha} (1 + d^I_{1,m})^{\alpha}} \right)}_{\mathcal L_{D_2}}
					\prod_{m\! =\! 1}^M \! \prod_{\substack{m' \! \neq\!  m,\\ m'\! =\! 1} }^M \! \!\! \! \underbrace{{\mathcal L}_{\epsilon_2}\! \! \left(\!\! \frac{ -k \varpi (1 \! \!+\!\!  d^D_{1,1})^{\alpha} \beta_{1,m,1} \beta_{1,m',1} }{ (1 \! \!+ \!\! d^R_{m,1})^{\frac{\alpha}{2}} (1 \! \!+\!\!  d^I_{1,m})^{\frac{\alpha}{2}}  (1\! \! +\!\!  d^R_{m',1})^{\frac{\alpha}{2}} (1 \!\! + \!\! d^I_{1,m'})^{\frac{\alpha}{2}} }  \! \!\right)}_{\mathcal L_{\epsilon_2}} \\
					&   \prod_{n'=2}^N \prod_{m=1}^M \underbrace{{\mathcal L}_{I_2} \left( \frac{ k \varpi \beta^2_{n',m,1} (1 + d^D_{1,1})^{\alpha} \delta  }{ ( 1 + d^I_{n',m})^{\alpha}( 1 + d^R_{m,1} )^{\alpha}} \right)}_{\mathcal L_{I_2}}
					 \Bigg] \\
					\!\overset{(a)}{ \approx}\!\! &
				\int_0^{\infty}{\rm exp}\Bigg(-  \frac{k \varpi(1 + d^D_{1,1})^{\alpha}\delta\sigma_n^2}{P_{tr}} \Bigg)
				{\rm exp}\Bigg( \!\!\!-\!{\alpha}\pi^2\!\lambda_m\lambda'_{un} \!\!\int_0^{\infty}\!\!\!\!\! \int_0^{\infty} \!\!\!\!\!\!(1\!\!-\!\!\mathcal L_{I_2} \mathcal L_{\epsilon_2}) d^R_{m,1} d^D_{n',1} {\text d} d^R_{m,1} {\text d} d^{D}_{n',1} \!\Bigg)\\
				  &
				  {\rm exp}\Bigg( -2\pi\lambda_m \int_0^{\infty} (1 - \mathcal L_{\epsilon_1} \mathcal L_{D_2}  ) d^R_{m,1} {\text d} d^R_{m,1} \Bigg)
				{\rm exp} \Bigg( -2\pi \lambda'_{un}\int_0^{\infty}(1 - \mathcal L_{I_1}) d^D_{n',1} {\text d} d^D_{n',1}  \Bigg) f(d^D_{1,1}) {\text d}d^D_{1,1},
			\end{aligned}\label{cover_e}
			\end{equation}
		\end{footnotesize}
	\end{figure*}%
	\begin{figure*}[t]
		\begin{footnotesize}
			\begin{subequations}
				\begin{align}
					{\mathcal M}_I(z) \notag
					= 
					&  {\rm E}\Bigg[  \prod_{n'=2}^N \prod_{m=1}^M \underbrace{{\mathcal L}_{I_2} \left( \frac{ z  \beta^2_{n',m,1} P_{tr} }{ ( 1 + d^I_{n',m})^ \alpha( 1 + d^R_{m,1} )^ \alpha} \right) }_{\mathcal L_{I_2} (z)}   
					    \prod_{n'=2}^N \underbrace{{\mathcal L}_{I_1}  \left( \frac{ z P_{tr}}{ (  1+d^D_{n',1} )^ \alpha} \right)}_{\mathcal L_{I_1}(z)}    \Bigg]  \tag{31a} \label{MI}\\
					= 
					& \int_0^{\infty} \!\!{\rm exp}\Bigg( \!\!-\! 4\pi^2\!\lambda_m\lambda'_{un} \!\!\! \iint \!\!(1\!\!-\!\!\mathcal L_{I_2}(z)) d^R_{m,1} d^D_{n',1} {\text d} d^R_{m,1} {\text d} d^{D}_{n',1} \!\Bigg) \notag 
					{\rm exp} \Bigg( -2\pi \lambda'_{un}\int(1 - \mathcal L_{I_1}(z) ) d^D_{n',1} {\text d} d^D_{n',1}  \Bigg) f(d^D_{1,1}) {\text d}d^D_{1,1},  \notag \\
					{\mathcal M}_{I + D}(z) \notag 
				   \approx & 
				   {\rm E} \! \Bigg[ \underbrace{{\mathcal L}_{D_1}\!\!\left(\!\frac{-zP_{tr}}{(1 \!+\! d^D_{1,1})^ \alpha}\!\right) }_{\mathcal L_{D_1}(z)} \!\prod_{n'=2}^N  \!\! \underbrace{{\mathcal L}_{I_1}\!\! \left( \!\frac{ z P_{tr} }{ (  1\!+\!d^D_{n',1} )^ \alpha} \!\right)}_{\mathcal L_{I_1}(z)}  
				   \!\prod_{m=1}^M \! \underbrace{{\mathcal L}_{\epsilon_1}\! \left(  \frac{   z 2 \Gamma(\varsigma+1) \beta_{1,m,1} P_{tr}  }{ \sqrt{\varsigma}\Gamma(\varsigma) (1 + d^D_{1,1})^{\frac{\alpha}{2}}(1 + d^R_{m,1})^{\frac{\alpha}{2}} (1 + d^I_{1,m})^{\frac{\alpha}{2}} }  \right)}_{\mathcal L_{\epsilon_1}(z)}   \tag{31b} \label{MID} \\ 
				   &\prod_{m=1}^{M} \underbrace{{\mathcal L}_{D_2} \left( \frac{ z\beta^2_{1,m,1} P_{tr}  }{  (1 + d^R_{m,1})^ \alpha (1 + d^I_{1,m})^ \alpha} \right)}_{\mathcal L_{D_2}(z)} \notag
				   \prod_{m\! =\! 1}^M \! \prod_{\substack{m' \! \neq\!  m,\\ m'\! =\! 1} }^M \! \!\! \! \underbrace{{\mathcal L}_{\epsilon_2}\! \! \left(\!\! \frac{ z\beta_{1,m,1} \beta_{1,m',1}  P_{tr} }{ (1 \! \!+ \!\! d^R_{m,1})^{\frac{\alpha}{2}} (1 \! \!+\!\!  d^I_{1,m})^{\frac{\alpha}{2}}  (1\! \! +\!\!  d^R_{m',1})^{\frac{\alpha}{2}} (1 \!\! + \!\! d^I_{1,m'})^{\frac{\alpha}{2}} }  \! \!\right)}_{\mathcal L_{\epsilon_2}(z)} \notag \\
				      & \prod_{n'=2}^N \prod_{m=1}^M \underbrace{{\mathcal L}_{I_2} \left( \frac{ z\beta^2_{n',m,1} P_{tr}  }{ ( 1 + d^I_{n',m})^ \alpha( 1 + d^R_{m,1} )^ \alpha} \right)}_{\mathcal L_{I_2}(z)}
					\Bigg] \notag \\
				   = & \int_0^{\infty} \!\! \mathcal L_{D_1}(z) {\rm exp}\Bigg( -2\pi\lambda_m \int (1 - \mathcal L_{\epsilon_1}(z) \mathcal L_{D_2}(z)  ) d^R_{m,1} {\text d} d^R_{m,1} \Bigg) 
				   {\rm exp} \Bigg( -2\pi \lambda'_{un}\int(1 - \mathcal L_{I_1}(z)) d^D_{n',1} {\text d} d^D_{n',1}  \Bigg)
				     \notag \\
				   & {\rm exp}\Bigg( \!\!-\! 4\pi^2\!\lambda_m\lambda'_{un} \!\!\! \iint \!\!(1\!\!-\!\!\mathcal L_{I_2}(z) {\mathcal L}_{\epsilon_2}(z) ) d^R_{m,1} d^D_{n',1} {\text d} d^R_{m,1} {\text d} d^{D}_{n',1} \!\Bigg) f(d^D_{1,1}) {\text d}d^D_{1,1} .\notag
								\end{align}
			\end{subequations}
			\end{footnotesize}
			\hrulefill
	\end{figure*}

	As aforementioned, the expectation of $D_{2,\Psi^M_{1}} $ and $D_{2,\widetilde{  \Psi}^M_{1}}$ are respectively expressed as
	\setcounter{equation}{26}
	\begin{small}
		\begin{subequations}
			\begin{align}
				{\rm E}  \! \left[ | D_{2,\Psi^M_{1}} |^2 \right]
				\!  \! \overset{(a)}{=} \! &  \left[ \int_0^{\infty} \left(    1  -  \frac{\Gamma^4( \varsigma +\frac{1}{2})}{\varsigma^2 \Gamma^4(\varsigma)}  +  \frac{Q\Gamma^4( \varsigma+\frac{1}{2})}{\varsigma^2 \Gamma^4(\varsigma)}    \right) \right. \notag \\
				& \cdot  2 \pi \lambda_m     Q P_{tr}  \mathcal Q_{p}(\alpha, d^D_{1,1}) +   4 \pi^2 \lambda_m^2  Q^2  P_{tr}  \\
				&  \left.  \cdot \frac{ \Gamma^4(\varsigma + \frac{1}{2}) }{\varsigma^2 \Gamma^4(\varsigma)}    \mathcal Q^2_{p}(\frac{\alpha}{2}, d^D_{1,1}) \right]  f(d^D_{1,1}) {\text d} d^D_{1,1}  ,\notag
				\label{D_2_beam}\\
				{\rm E} \left[ | D_{2,\widetilde{\Psi}^M_{1}} |^2 \right]
				\!  \!  \overset{(b)}{=}  &  
				  2 \pi \lambda_m Q  P_{tr} \!\!\int_0^{\infty} \!\!\!\! {\mathcal Q}_{1-p}(\alpha, d^D_{1,1}) f(d^D_{1,1}) {\text d} d^D_{1,1} ,
			\end{align}
		\end{subequations}
	\end{small}%
	where $(a)$ and $(b)$ follows the Campbell's theorem in ternary stochastic geometry.
	Thereafter,  we deduce the average power of the entire signal $D_{sum}$, as shown in \eqref{x_mean}. 
	
	Recall Section II and III, the average power of interference $I_1$ and $I_2$ are respectively expressed as
	\begin{small}
		\begin{subequations}
			\begin{align}
				{\rm{E}} \!\left[ \left|I_1 \right|^2 \right] \!\!\!
				\overset{(a)}{=} & 2 \pi      \lambda_{un}' P_{tr}  \left( \frac{1}{6} - \mathcal{F}\left( \alpha, \pi \lambda_{un}' \right) \right), \\
				{\rm{E}} \!\left[ \left|I_2 \right|^2 \right]\!\!\!
			\overset{(b)}{=} &    Q P_{tr} {\rm{E}} \!\!\left[  \! \sum_{ \substack{ n'=2}}^{N}  \!\! \frac{2 \pi \lambda_m  \int_0^{\infty} \!\! {\mathcal G}\!\!\left( \!\! \frac{1}{(1\!+\!d^R_{m,1})^{\alpha} } \!\!\right) d_{m,1}^{R} {\text d} d_{m,1}^{R}}{(1 + d_{n',m}^{I}  )^{  \alpha}}  \right] \\
			= & 4 \pi^2 \lambda_{un}' \lambda_m  Q  P_{tr}  \! \!\! \iint \!\!\! \mathcal Q_{1}(\alpha, d^D_{1,1}) d_{n',1}^{D} {\text d} d_{n',1}^{D} f ( d_{1,1}^D ) {\text d}  d_{1,1}^D, \notag
			\end{align} 
		\end{subequations}
		\end{small}%
	where $(a)$ and $(b)$ both follow the Campbell's theorem. 
	Finally, the average power of the entire interference $I_{\rm sum}$ is deduced as \eqref{I_mean}.

	\section{}

	According to \eqref{coverage_hebing}, ${\rm E}\left[  {\rm exp}\Big(  - k \varpi \mathcal X   \Big)   \right]$ is expressed as \eqref{cover_e}, where $(a)$ is based on Lemma \ref{lemma_2} and $\beta_{n,m,u}$ can be integrated using ${\mathcal{G}}(x) \triangleq \int_0^{\pi} \frac{{\rm P}(\beta_{n,m,1} = 1 | \Delta \theta^M_{n,m,1} ) }{\pi} x {\text d} \Delta \theta$.

	\section{} 
	The ASE is defined as the sum rate per bandwidth per area, expressed as
	\setcounter{equation}{29}
	\begin{small}
	\begin{equation}
		\begin{aligned}
		\label{ASE_UDN_Appendix}
		\eta_s 
		=    \frac{\lambda'_{un}  }{{\text ln}2} \int_{0}^{\infty}[{\mathcal M}_I(z) - {\mathcal M}_{I+D}(z) ]\frac{{\rm exp}({-z\sigma_n^2 }) }{z} {\text d} z,
		\end{aligned}
	\end{equation}
	\end{small}
	where ${\mathcal M}_x(z)= {\rm E} \left[  e^{-zx} \right]$ is the moment generation function (MGF) of $x$. 
	According to {Lemma 2}, ${\mathcal M}_I(z)$ and ${\mathcal M}_{I+D}(z) $ are given as \eqref{MI} and \eqref{MID}, respectively.
	Thus, the ASE of the RIS-assisted system is shown in \eqref{ase_frist}.

	\end{appendices}

\end{document}